%% file: arXiv.tex
\RequirePackage{silence}
\documentclass[a4paper]{scrartcl}
\usepackage[UKenglish]{babel}
\usepackage[utf8]{inputenc}

\usepackage{microtype}
\usepackage[UKenglish]{babel}
\usepackage[libertine]{newtxmath}
\usepackage[tt=false, type1=true]{libertine}
\usepackage{csquotes}
\usepackage{geometry}
\usepackage{ifthen}
\newif\iflong
\longtrue

\usepackage[anythingbreaks]{breakurl}

\usepackage{amsmath,amssymb}
\usepackage{stmaryrd}
\usepackage[standard,amsmath,thmmarks]{ntheorem}

\usepackage{multicol}
\usepackage{comment}

\usepackage[
    backend=biber,
    bibencoding=auto,
    style=numeric,
    sortlocale=en_US,
    maxbibnames=200,
    url=true, 
    doi=true,
    eprint=false    
]{biblatex}
\makeatletter
\renewtheoremstyle{plain}%
{\item[\hskip\labelsep \theorem@headerfont ##1\ ##2 \theorem@separator]}%
{\item[\hskip\labelsep \theorem@headerfont ##1\ ##2\ \normalfont({\sffamily##3})  \theorem@separator]}
\makeatother

\renewtheorem{definition}[theorem]{Definition}
\renewtheorem{lemma}[theorem]{Lemma}
\renewtheorem{proposition}[theorem]{Proposition}
\renewtheorem{corollary}[theorem]{Corollary}
\renewtheorem{example}[theorem]{Example}
\renewtheorem{theorem}{Theorem}

\RedeclareSectionCommand[indent=0pt]{subparagraph}

\usepackage{microtype}%if unwanted, comment out or use option "draft"

\title{Team Semantics for the Specification and Verification of Hyperproperties\footnote{This work was supported by the DFG projects ME4279/1-1 and ``TriCS'' (ZI 1516/1-1).}}

\author{Andreas Krebs$^1$ \and Arne Meier$^2$ \and Jonni Virtema$^3$ \and Martin Zimmermann$^4$}

\date{\small $^1$ Wilhelm-Schickard-Institut für Informatik, Universität Tübingen\\ 72076 Tübingen, Germany\\ \texttt{krebs@informatik.uni-tuebingen.de} \\[.3cm] 
$^2$ Institut für Theoretische Informatik, Leibniz Universität Hannover,\\ Appelstrasse 4, 30167 Hannover, Germany\\ \texttt{meier@thi.uni-hannover.de} \\[.3cm] 
$^3$ Hasselt University, 3590 Diepenbeek, Belgium\\ \texttt{jonni.virtema@uhasselt.be} \\[.3cm] 
	$^4$ Reactive Systems Group, Saarland University,\\ 66123 Saarbrücken, Germany\\ \texttt{zimmermann@react.uni-saarland.de}
}

\input{macros}
\input{teamltl-packages}

\renewcommand{\decisionproblem}[3]{
\begin{description}\itemsep0mm
	\item[Problem:] #1
	\item[Input:] #2 
	\item[Question:] #3?
\end{description}
}

\renewcommand{\hshift}{-3mm}
\renewcommand{\vshift}{1.75mm}
\renewcommand{\hshiftstar}{-3mm}

\begin{document}

\maketitle

\begin{abstract}
\section*{Abstract}
We develop team semantics for Linear Temporal Logic (LTL) to express hyperproperties, which have recently been identified as a key concept in the verification of information flow properties.
	Conceptually, we consider an asynchronous and a synchronous variant of team semantics.
	We study basic properties of this new logic and classify the computational complexity of its satisfiability, path, and model checking problem.
	Further, we examine how extensions of these basic logics react on adding other atomic operators.
	Finally, we compare its expressivity to the one of HyperLTL, another recently introduced logic for hyperproperties.
	Our results show that LTL under team semantics is a viable alternative to HyperLTL, which complements the expressivity of HyperLTL and has partially better algorithmic properties.\end{abstract}

\input{content.tex}

\printbibliography

\end{document}

%% file: macros.tex
\newcommand{\LTL}{\protect\ensuremath{\mathrm{LTL}}\xspace}
\newcommand{\CTL}{\protect\ensuremath{\mathrm{CTL}}\xspace}

\newcommand{\dfn}{\mathrel{\mathop:}=}

\newcommand{\ltlopFont}[1]{\mathsf{#1}}
\newcommand{\X}{\ltlopFont{X}}
\newcommand{\F}{\ltlopFont{F}}
\newcommand{\G}{\ltlopFont{G}}
\newcommand{\U}{\ltlopFont{U}}
\newcommand{\R}{\ltlopFont{R}}

\newcommand{\hshift}{-3mm}
\newcommand{\vshift}{1.5mm}
\newcommand{\hshiftstar}{-2.9mm}
\newcommand\nhyltlmodels{\,\mbox{$\not\models\hspace{\hshift}\raisebox{\vshift}{\scriptsize h}\,\,$}\,}
\newcommand\hyltlmodels{\,\mbox{$\models\hspace{\hshift}\raisebox{\vshift}{\scriptsize h}\,\,$}\,}
\newcommand\ltlmodels{\,\mbox{$\models\hspace{\hshift}\raisebox{\vshift}{\scriptsize $c$}\,$}\,} 

\newcommand\amodels{\,\mbox{$\models\hspace{\hshift}\raisebox{\vshift}{\scriptsize a}\,$}\,}
\newcommand\smodels{\,\mbox{$\models\hspace{\hshift}\raisebox{\vshift}{\scriptsize s}\,$}\,}
\newcommand\starmodels{\,\mbox{$\models\hspace{\hshiftstar}\raisebox{\vshift}{$\scriptscriptstyle \star$}\,$}\,}
\newcommand\nstarmodels{\,\mbox{$\not\models\hspace{\hshiftstar}\raisebox{\vshift}{$\scriptscriptstyle \star$}\,$}\,}

\newcommand\nsmodels{\,\mbox{$\not\models\hspace{\hshift}\raisebox{\vshift}{\scriptsize s}\,$}\,}
\newcommand{\N}{\mathbb N}
\newcommand{\ap}{\mathrm{AP}}

\newcommand{\TPCa}{\mathrm{TPC}^{\mathrm{a}}}
\newcommand{\TPCs}{\mathrm{TPC}^{\mathrm{s}}}
\newcommand{\TMCa}{\mathrm{TMC}^{\mathrm{a}}}
\newcommand{\TMCs}{\mathrm{TMC}^{\mathrm{s}}}
\newcommand{\TSAT}{\mathrm{TSAT}}
\newcommand{\TSATa}{\mathrm{TSAT}^{\mathrm{a}}}
\newcommand{\TSATs}{\mathrm{TSAT}^{\mathrm{s}}}
\newcommand{\LTLMC}{\mathrm{LTL\text-MC}}
\newcommand{\LTLPC}{\mathrm{LTL\text-PC}}

\newcommand{\leqcd}{\leq_{\mathrm{cd}}}
\newcommand{\leqpm}{\leq_{\mathrm{m}}^{\mathrm{p}}}
 
\newcommand{\PSPACE}{\protect\ensuremath{\mathrm{PSPACE}}} 

\newcommand{\ATIME}{\protect\ensuremath{\mathrm{ATIME\text-ALT}}}

\newcommand{\pol}{\protect\ensuremath{\mathrm{pol}}}
\newcommand{\EXPTIME}{\protect\ensuremath{\mathrm{EXPTIME}}}
\newcommand{\NEXPTIME}{\protect\ensuremath{\mathrm{NEXPTIME}}}
\newcommand{\PL}{\protect\ensuremath{\mathrm{PL}}}
 
\newcommand{\PTime}{\protect\ensuremath{\mathrm{P}}}
\newcommand{\NC}[1]{\protect\ensuremath{\mathrm{NC}^{#1}}}
\newcommand{\AC}[1]{\protect\ensuremath{\mathrm{AC}^{#1}}}
\newcommand{\LOGDCFL}{\protect\ensuremath{\mathrm{logDCFL}}}

\newcommand{\ra}{\mathrm{a}}
\newcommand{\rs}{\mathrm{s}}

\newcommand{\pow}[1]{2^{#1}}
\newcommand{\size}[1]{|#1|}
\newcommand{\set}[1]{\{#1\}}
\renewcommand{\phi}{\varphi}
\renewcommand{\epsilon}{\varepsilon}
\newcommand{\nats}[0]{\mathbb{N}}

\newcommand{\cceq}{\mathop{::=}}
\newcommand{\suffix}[2]{#1[#2,\infty)}
\newcommand{\var}{\mathcal{V}}
\newcommand{\hyltl}{\protect\ensuremath{\mathrm{HyperLTL}}\xspace}
\DeclareMathOperator\dep{\mathrm{dep}}

\newcommand{\kripke}{\mathcal{K}}
\newcommand{\aut}{\mathcal{A}}

\newcommand{\decisionproblem}[3]{
\begin{samepage}
\begin{description}
	\item[Problem:] #1
	\item[Input:] #2 
	\item[Question:] #3?
\end{description}
\end{samepage}
}

\newcommand{\myquot}[1]{``#1''}

\newcommand{\yes}{\protect\ensuremath{\checkmark}}
\newcommand{\no}{\protect\ensuremath{\times}}

\newcommand{\teamup}{\mathcal{T}}
\newcommand{\eval}[1]{\llbracket #1\rrbracket}
\newcommand{\lcm}{\mathrm{lcm}}
\newcommand{\prfx}{\mathrm{prfx}}

\newcommand{\qbfval}{\mathrm{QBF\text-VAL}}

\newcommand{\dom}{\mathrm{Dom}}

\newcommand{\infull}{The details can be found in the full version \cite{DBLP:journals/corr/abs-1709-08510}.}

%% file: teamltl-packages.tex
\usepackage[disable]{todonotes}
\usepackage[utf8]{inputenc}

\usepackage{xspace}
\usepackage{mathtools,amssymb}

\usepackage{floatflt}

\usepackage{booktabs}
\usepackage{stmaryrd}
\usepackage{xspace}
\usepackage[ruled,vlined]{algorithm2e}

\usetikzlibrary{arrows}
\usetikzlibrary{decorations.pathreplacing}

\tikzset{% 
dot/.style={fill,black,circle,inner sep=0mm,minimum width=2mm},
dotwhite/.style={draw,black,fill=white,circle,inner sep=0mm,minimum width=2mm},
dw/.style={draw,black,fill=white,circle,inner sep=0mm,minimum width=3.5mm},
d/.style={draw,black,fill=black,circle,inner sep=0mm,minimum width=3.5mm,text=white}
}

%% file: content.tex
\section{Introduction}
Guaranteeing security and privacy of user information is a key requirement in software development. 
However, it is also one of the hardest goals to accomplish. 
One reason for this difficulty is that such requirements typically amount to reasoning about the flow of information and relating different execution traces of the system. 
In particular, these requirements are no longer trace properties, i.e., properties whose satisfaction can be verified by considering each trace in isolation. 
For example, the property~\myquot{the system terminates eventually} is satisfied if every trace eventually reaches a final state. 
Formally, a trace property~$\phi$ is a set of traces and a system satisfies $\phi$ if each of its traces is in $\phi$. 

In contrast, the property~\myquot{the system terminates within a bounded amount of time} is no longer a trace property; consider a system that has a trace~$t_n$ for every $n$, so that $t_n$ only reaches a final state after $n$ steps. 
This system does not satisfy the bounded termination property, but each individual trace~$t_n$ could also stem from a system that does satisfy it. 
Thus, satisfaction of the property cannot be verified by considering each trace in isolation. 

Properties with this characteristic were termed \emph{hyperproperties} by Clarkson and Schnei\-der~\cite{DBLP:journals/jcs/ClarksonS10}. 
Formally, a hyperproperty~$\phi$ is a set of sets of traces and a system satisfies $\phi$ if its set of traces is contained in $\phi$. 
The conceptual difference to trace properties allows specifying a much richer landscape of properties including information flow and trace properties.
Further, one can also express specifications for symmetric access to critical resources in distributed protocols and Hamming distances between code words in coding theory~\cite{markusPhD}. 
However, the increase in expressiveness requires novel approaches to specification and~verification. 

\subparagraph*{HyperLTL} 
Trace properties are typically specified in temporal logics, most prominently in Linear Temporal Logic (\LTL)~\cite{Pnueli/1977/TheTemporalLogicOfPrograms}. 
Verification of \LTL specifications is routinely employed in industrial settings and marks one of the most successful applications of formal methods to real-life problems. 
Recently, this work has been extended to hyperproperties: \hyltl, \LTL equipped with trace quantifiers, has been introduced to specify hyperproperties~\cite{DBLP:conf/post/ClarksonFKMRS14}. 
Accordingly, a model of a \hyltl formula is a set of traces and the quantifiers range over these traces. 
This logic is able to express the majority of the information flow properties found in the literature (we refer to Section~3 of~\cite{DBLP:conf/post/ClarksonFKMRS14} for a full list). 
The satisfiability problem for \hyltl is undecidable~\cite{finkbeiner_et_al:LIPIcs:2016:6170} while the model checking problem is decidable, albeit of non-elementary complexity~\cite{DBLP:conf/post/ClarksonFKMRS14,DBLP:conf/cav/FinkbeinerRS15}. 
In view of this, the full logic is too strong. Fortunately most information flow properties found in the literature can be expressed with at most one quantifier alternation and consequently belong to decidable (and tractable) fragments. 
Further works have studied runtime verification~\cite{DBLP:conf/rv/BonakdarpourF16,DBLP:conf/rv/FinkbeinerHST17}, connections to first-order logic~\cite{DBLP:conf/stacs/Finkbeiner017}, provided tool support~\cite{DBLP:conf/cav/FinkbeinerRS15,finkbeiner_et_al:LIPIcs:2016:6170}, and presented applications to \myquot{software doping}~\cite{DBLP:conf/esop/DArgenioBBFH17} and the verification of web-based workflows~\cite{634}. 
In contrast, there are natural properties, e.g., bounded termination, which are not expressible in \hyltl (which is an easy consequence of a much stronger non-expressibility result~\cite{BozzelliMP15}).

\subparagraph*{Team Semantics} Intriguingly, there exists another modern family of logics, \emph{Dependence Logics}~\cite{vaananen07, DKV16}, which operate as well on sets of objects instead of objects alone. 
Informally, these logics extend first-order logic (FO) by atoms expressing, e.g., that \myquot{the value of a variable~$x$ functionally determines the value of a variable~$y$} or that \myquot{the value of a variable~$x$ is informationally independent of the value of a variable~$y$}. 
Obviously, such statements only make sense when being evaluated over a set of assignments. 
In the language of dependence logic, such sets are called \emph{teams} and the semantics is termed \emph{team semantics}.

In 1997, Hodges introduced compositional semantics for Hintikka's Independence-friendly logic~\cite{Hodges97c}.
This can be seen as the cornerstone of the mathematical framework of dependence logics. 
Intuitively, this semantics allows for interpreting a team as a database table.
In this approach, variables of the table correspond to attributes and assignments to rows or records.
In 2007, Väänänen~\cite{vaananen07} introduced his modern approach to such logics and adopted team semantics as a core notion, as dependence atoms are meaningless under Tarskian semantics.

After the introduction of dependence logic, a whole family of logics with different atomic statements have been introduced in this framework: \emph{independence logic}~\cite{gv13} and \emph{inclusion logic}~\cite{Galliani12} being the most prominent.
Interest in these logics is rapidly growing and the research community aims to connect their area to a plethora of disciplines, e.g., linguistics~\cite{dagstuhl15}, biology~\cite{dagstuhl15}, game~\cite{bradfield15}  and social choice theory~\cite{shpitser15}, philosophy~\cite{shpitser15}, and computer science~\cite{dagstuhl15}.
We are the first to exhibit connections to formal languages via application of Büchi automata (see Theorem~\ref{thm:tmcs-splitfree}).
Team semantics has also found their way into modal \cite{va08} and temporal logic~\cite{kmv15}, as well as statistics~\cite{dhkmv16}.

Recently, Krebs et~al.~\cite{kmv15} proposed team semantics for Computation Tree Logic (CTL), where a team consists of worlds of the transition system under consideration. 
They considered synchronous and asynchronous team semantics, which differ in how time evolves in the semantics of the temporal operators. 
They proved that satisfiability is $\EXPTIME$-complete under both semantics while model checking is $\PSPACE$-complete under synchronous semantics and $\PTime$-complete under asynchronous semantics. 

\subparagraph*{Our Contribution}
The conceptual similarities between \hyltl and team semantics raise the question how an \LTL variant under {team semantics} relates to \hyltl.
For this reason, we develop team semantics for \LTL, analyse the complexity of its satisfiability and model checking problems, and subsequently compare the novel logic to \hyltl.

When defining the logic, we follow the approach of Krebs et~al.\ \cite{kmv15} for defining team semantics for \CTL: 
we introduce synchronous and asynchronous team semantics for \LTL, where teams are now sets of traces. 
In particular, as a result, we have to consider potentially uncountable teams, while all previous work on model checking problems for logics under team semantics has been restricted to the realm of finite teams. 

%Our results are as follows: % (see Figure~\ref{fig:overview}):
We prove that the satisfiability problem for team $\LTL$ is $\PSPACE$-complete under both semantics, by showing that the problems are equivalent to $\LTL$ satisfiability under classical semantics. 
Generally, we observe that for the basic asynchronous variant all of our investigated problems trivially reduce to and from classical $\LTL$ semantics.
However, for the synchronous semantics this is not the case for two variants of the model checking problem. 
As there are uncountably many traces, we have to represent teams, i.e., sets of traces, in a finitary manner. 
The path checking problem asks to check whether a finite team of ultimately periodic traces satisfies a given formula. 
As our main result, we establish this problem to be $\PSPACE$-complete for synchronous semantics. 
In the (general) model checking problem, a team is represented by a finite transition system. 
Formally, given a transition system and a formula, the model checking problem asks to determine whether the set of traces of the system satisfies the formula. 
For the synchronous case we give a polynomial space algorithm for the model checking problem for the disjunction-free fragment, while we leave open the complexity of the general problem.  
Disjunction plays a special role in team semantics, as it splits a team into two. 
As a result, this operator is commonly called \emph{splitjunction} instead of disjunction.
In our setting, the splitjunction requires us to deal with possibly infinitely many splits of uncountable teams, if a splitjunction is under the scope of a $\G$-operator, which raises interesting language-theoretic questions. 

Further, we study the effects for complexity that follow when our logics are extended by dependence atoms and the contradictory negation.
Finally, we show that \LTL under team semantics is able to specify properties which are not expressible in \hyltl and \emph{vice versa}.

Recall that satisfiability for \hyltl is undecidable and model checking of non-elementary complexity. 
Our results show that similar problems for \LTL under team semantics have a much simpler complexity while some hyperproperties are still expressible (e.g., input determinism, see page~\pageref{pg:input-nonint}, or bounded termination).
This proposes \LTL under team semantics to be a significant alternative for the specification and verification of hyperproperties that complements \hyltl.
%
%
%\begin{figure}
%	$$\begin{array}{lll}\toprule
%		\text{Problem} & \text{Complexity} & \text{Reference}\\\midrule
%		\TSATs(\dep)& \PSPACE & \text{Thm.~\ref{thm:TSATa-TSATs-dep}} \\
%		\TPCa(\dep)&\PSPACE\text{-hard} & \text{Thm.~\ref{thm:TPCa-dep}} \\		
%		\TPCs&\PSPACE & \text{Thm.~\ref{thm:tpcs}}\\
%		\quad+\dep&\PSPACE\text{-hard} & \text{above} \\
%		\quad+\text{ gen.~atoms}&\PSPACE & \text{Thm.~\ref{thm:tpcs-fo-npcomplete}} \\
%		\TMCa(\dep)&\NEXPTIME\text{-hard}&\text{Thm.~\ref{thm:TMCa-TMCs-dep}}\\
%		\TMCa(\sim)&\ATIME(\expo,\pol)\text{-hard}&\text{Thm.~\ref{thm:tmca-tmcs-negation}}\\
%		\TMCs& \PSPACE\text{-hard} & \text{Thm.~\ref{thm:tpcs}}\\
%		\quad-\,\lor & \in\PSPACE& \text{Thm.~\ref{thm:tmcs-splitfree}}\\
%		\quad+\dep&\NEXPTIME\text{-hard}&\text{Thm.~\ref{thm:TMCa-TMCs-dep}}\\
%		\quad+\sim&\ATIME(\expo,\pol)\text{-hard}&\text{Thm.~\ref{thm:tmca-tmcs-negation}}\\\bottomrule
%	\end{array}$$	
%	\caption{Overview of results. `$\dep{}$' refers to dependence atoms, `$\sim$' refers to negation, and `$+$'/`$-$' means added to/removed from the basic problem. `$\ATIME(x,y)$' means alternating time with a runtime of $x$ and $y$-many alternations. All results are completeness results unless otherwise stated.}\label{fig:overview}
%\end{figure}
%

\section{Preliminaries}

The non-negative integers are denoted by $\nats$ and the power set of a set~$S$ is denoted by $\pow{S}$. Throughout the paper, we fix a finite set~$\ap$ of atomic propositions.

\subparagraph*{Computational Complexity}
We will make use of standard notions in complexity theory. 
In particular, we will use the complexity classes $\PTime$ and $\PSPACE$.
Most reductions used in the paper are $\leqpm$-reductions, that is, polynomial time, many-to-one reductions.

\subparagraph*{Traces}
 A \emph{trace} over $\ap$ is an infinite sequence from $ (\pow{\ap})^\omega$; a finite trace is a finite sequence from $(\pow{\ap})^*$. 
 The length of a finite trace~$t$ is denoted by $\size{t}$. The empty trace is denoted by $\epsilon$ and the concatenation of two finite traces~$t_0$ and $t_1$ by $t_0t_1$.
 Unless stated otherwise, a trace is always assumed to be infinite. 
 A \emph{team} is a (potentially infinite) set of traces.
 
  Given a trace~$t = t(0) t(1) t(2) \cdots$ and $i \ge 0$, we define $t[i,\infty) \dfn t(i) t(i+1) t(i+2) \cdots$, which we lift to teams~$T \subseteq (\pow{\ap})^\omega$ by defining $T[i,\infty) \dfn \set{ t[i,\infty) \mid t \in T }$. 
 A trace~$t$ is \emph{ultimately periodic}, if it is of the form~$t = t_0 \cdot t_1^\omega = t_0 t_1 t_1 t_1 \cdots$ for two finite traces~$t_0$ and $t_1$ with $\size{t_1}>0$. 
 As a result, an ultimately periodic trace $t$ is finitely represented by the pair~$(t_0, t_1)$; we define $\eval{(t_0, t_1)} = t_0t_1^\omega$. 
 Given a set~$\teamup$ of such pairs, we define $\eval{\teamup} = \set{\eval{(t_0,t_1)} \mid (t_0, t_1) \in \teamup}$, which is a team of ultimately periodic traces. 
 We call $\teamup$ a team encoding of $\eval{\teamup}$. 

\subparagraph*{Linear Temporal Logic}
The formulas of Linear Temporal Logic (\LTL)~\cite{Pnueli/1977/TheTemporalLogicOfPrograms} are defined via the grammar
$
\varphi \cceq 
p\mid \lnot p\mid\varphi\land\varphi\mid\varphi\lor\varphi \mid \X\varphi\mid 
\F\varphi\mid \G\varphi\mid 
\varphi\U\varphi \mid \varphi\R\varphi,
$
where $p$ ranges over the atomic propositions in $\ap$.
The length of a formula is defined to be the number of Boolean and temporal connectives occurring in it. The length of an \LTL formula is often defined to be the number of syntactically different subformulas, which might be exponentially smaller. 
Here, we need to distinguish syntactically equal subformulas which becomes clearer after defining the semantics (see also Example~\ref{example_semantics} afterwards on this).
As we only consider formulas in negation normal form, we use the full set of temporal operators.

%%%%%%%% remove from here on
Next, we recall the classical semantics of \LTL before we introduce team semantics. 
For traces~$t \in (\pow{\ap})^\omega$ we define
\begin{multicols}{2}
\begin{tabbing}
    $t \ltlmodels  \psi\lor\phi$ \= if \= Rechts \kill
	$t \ltlmodels  p$\> if \>$p\in t(0)$,\\
	$t \ltlmodels  \lnot p$\> if \>$ p\notin t(0)$,\\
	$t \ltlmodels  \psi\land\phi$\> if \>$ t \ltlmodels \psi \text{ and }t \ltlmodels \phi$,\\
	$t \ltlmodels  \psi\lor\phi$\> if \>$ t \ltlmodels \psi \text{ or }t \ltlmodels \phi$,\\
	$t \ltlmodels \X\varphi$\> if \>$t[1,\infty)\ltlmodels\varphi$,\\
	$t\ltlmodels \F\varphi$\> if \>$\exists k \ge 0: t[k,\infty)\ltlmodels \varphi$,\\
	$t\ltlmodels \G\varphi $\> if\> $\forall k\ge 0: t[k,\infty)\ltlmodels \varphi$,\\
	$t\ltlmodels \psi\U\phi$\> if\> $\exists k \ge 0: t[k,\infty)\ltlmodels \phi$ and\\
	\>\quad\quad$\forall k' < k: t[k',\infty)\ltlmodels \psi$,\\ 
	$t\ltlmodels \psi\R\phi$\> if\> $\forall k\ge 0: t[k,\infty)\ltlmodels \phi$ or\\
	\>\quad\quad$\exists k' < k: t[k',\infty)\ltlmodels \psi$.
\end{tabbing}
\end{multicols}
%%% till here
\subparagraph*{Team Semantics for LTL}
Next, we introduce two variants of team semantics for \LTL, which differ in their interpretation of the temporal operators: a synchronous semantics ($\smodels$), where time proceeds in lockstep along all traces of the team, and an asynchronous semantics ($\amodels$) in which, on each trace of the team, time proceeds independently. We write $\starmodels$ whenever a definition coincides for both semantics. 
For teams $T \subseteq (\pow{\ap})^\omega$ let
\begin{multicols}{2}
\begin{tabbing}
    $T\starmodels \psi\land\phi$ \= if \= Rechts \kill
	$T\starmodels p$\> if \> $\forall t \in T: p\in t(0)$,\\
	$T\starmodels \lnot p$\> if \> $\forall t \in T: p\notin t(0)$,\\
	$T\starmodels \psi\land\phi$\> if \>$T\starmodels\psi$  and $T\starmodels\phi$,\\
	$T\starmodels \psi\lor\phi$\> if \> $\exists T_1\cup T_2=T: T_1\starmodels\psi$ and $T_2\starmodels\phi$,\\
	$T\starmodels\X\varphi$\> if \>$T[1,\infty)\starmodels\varphi$.
\end{tabbing}
\end{multicols}
This concludes the cases where both semantics coincide. Next, we present the remaining cases for the synchronous semantics, which are inherited from the classical semantics of \LTL.

\begin{tabbing}
	$T\smodels\psi\U\phi$ \= if \= Rechts\kill
	$T\smodels\F\phi$\> if \> $\exists k \ge 0: T[k,\infty)\smodels\phi$,\\
	$T\smodels\G\phi$\> if \> $\forall k \ge 0: T[k,\infty)\smodels\phi$,\\
	$T\smodels\psi\U\phi$\> if \> $\exists k\ge 0: T[k,\infty)\smodels\phi$ and $\forall k' <k: T[k',\infty)\smodels\psi$, and\\
	$T\smodels\psi\R\phi$\> if \> $\forall k \ge 0: T[k,\infty)\smodels\phi$ or $\exists k'<k: T[k',\infty)\smodels\psi$.
\end{tabbing}

Finally, we present the remaining cases for the asynchronous semantics. 
Note that, here there is no unique timepoint~$k$, but a timepoint $k_t$ for every trace~$t$, i.e., time evolves asynchronously between different traces.

%\begin{comment}
%\begin{tabbing}
%	$T\amodels\psi\U\phi$ \= if \= Rechts \kill
%	$T\amodels\F\phi$\> if \> $\forall t\in T\,\exists k_t \ge 0: \set{t[k_t,\infty) \mid t \in T}\amodels\phi$,\\
%	$T\amodels\G\phi$\> if \> $\forall t\in T\,\forall k_t\ge 0: \set{t[k_t,\infty) \mid t \in T}\amodels\phi$,\\
%	$T\amodels\psi\U\phi$\> if \> $\forall t\in T\,\exists k_t\ge 0: \set{t[k_t,\infty) \mid t \in T}\amodels\phi$ and\\
%	\>\> $\forall k_t' < k_t: \set{t[k_t',\infty) \mid t \in T}\amodels\psi$, and\\
%	$T\amodels\psi\R\phi$\> if \> $\forall t\in T\,\forall k_t \ge 0: \set{t[k_t,\infty) \mid t \in T}\amodels\phi$ or\\
%	\>\> $\exists k_t' \le  k_t: \set{t[k_t',\infty) \mid t \in T}\amodels\psi$.
%\end{tabbing}
%\end{comment}

\begin{tabbing}
	$T\amodels\psi\U\phi$ \= if \= Rechts \kill
	$T\amodels\F\phi$\> if \> $\exists k_t \ge 0$, for each  $t \in T$: $\set{t[k_t,\infty) \mid t \in T}\amodels\phi$\\
	$T\amodels\G\phi$\> if \> $\forall k_t\ge 0$, for each $t\in T$ : $\set{t[k_t,\infty) \mid t \in T}\amodels\phi$,\\
	$T\amodels\psi\U\phi$\> if \> $\exists k_t\ge 0$, for each $t\in T$ :  $\set{t[k_t,\infty) \mid t \in T}\amodels\phi$, and\\
	\>\> \quad\quad$\forall k_t' < k_t$,  for each $t\in T$  : $\set{t[k_t',\infty) \mid t \in T}\amodels\psi$, and\\
	$T\amodels\psi\R\phi$\> if \> $\,\forall k_t \ge 0$,  for each $t\in T$ : $\set{t[k_t,\infty) \mid t \in T}\amodels\phi$ or\\
	\>\> \quad\quad$\exists k_t' <  k_t$,  for each $t\in T$: $\set{t[k_t',\infty) \mid t \in T}\amodels\psi$.
\end{tabbing}

We call expressions of the form $\psi\lor\phi$ \emph{splitjunctions} to emphasise on the team semantics where we split a team into two parts.
Similarly, the $\lor$-operator is referred to as a \emph{splitjunction}.

Let us illustrate the difference between synchronous and asynchronous semantics with an example involving the $\F$ operator. Similar examples can be constructed for the other temporal operators (but for $\X$) as well. 

\begin{example}
\label{example_semantics}
Let $T = \set{\set{p} \emptyset^\omega,\emptyset\set{p} \emptyset^\omega}$. We have that $T \amodels \F p$, as we can pick $k_t = 0$ if $t = \set{p} \emptyset^\omega$, and $k_t = 1$ if $t = \emptyset\set{p} \emptyset^\omega$. On the other hand, there is no single~$k$ such that $T[k,\infty) \smodels p$, as the occurrences of $p$ are at different positions. Consequently $T \nsmodels \F p$.

Moreover, consider the formula $\F p\vee\F p$ which is satisfied by $T$ on both semantics. 
However, $\F p$ is not satisfied by $T$ under synchronous semantics.
Accordingly, we need to distinguish the two disjuncts $\F p$ and $\F p$ of $\F p\vee\F p$ to assign them to different teams.
\end{example}

\begin{figure}
	\begin{center}
\scalebox{0.85}{
 \begin{tabular}{llcc}\toprule
 	property & definition & $\amodels$ & $\smodels$\\\midrule
 	empty team property & $\emptyset\starmodels\phi$ & \yes & \yes\\
 	downwards closure & $T\starmodels\phi$ implies $\forall T'\subseteq T$: $T'\starmodels\phi$ & \yes & \yes \\
 	union closure & $T\starmodels\phi,\ T'\starmodels\phi$ implies $T\cup T'\starmodels\phi$ & \yes & \no\\
 	flatness & $T\starmodels\phi$ if and only if $\forall t\in T$: $\{t\}\starmodels\phi$	& \yes & \no\\
 	singleton equivalence & $\{t\}\starmodels\phi$ if and only if $t\ltlmodels\phi$ &\yes & \yes \\\bottomrule
 \end{tabular}
}
\end{center}
\caption{Structural properties overview.}\label{fig:structuralpropertiesoverview}
\end{figure}

In contrast, synchronous satisfaction implies asynchronous satisfaction, i.e., $T \smodels \phi$ implies $T \amodels \phi$. The simplest way to prove this is by applying downward closure, singleton equivalence, and flatness (see Fig.~\ref{fig:structuralpropertiesoverview}). Example~\ref{example_semantics} shows that the converse does not hold.%, i.e., asynchronous satisfaction does not imply synchronous satisfaction.

Next, we define the most important verification problems for \LTL in team semantics setting, namely satisfiability and two variants of the model checking problem: For classical \LTL, one studies the path checking problem and the model checking problem. The difference between these two problems lies in the type of structures one considers. Recall that a model of an \LTL formula is a single trace. In the path checking problem, a trace~$t$ and a formula~$\phi$ are given, and one has to decide whether $t \ltlmodels \phi$. This problem has applications to runtime verification and monitoring of reactive systems~\cite{kf09,ms03}. In the model checking problem, a Kripke structure~$\mathcal{K}$ and a formula~$\phi$ are given, and  one has to decide whether every execution trace~$t$ of $\mathcal{K}$ satisfies $\phi$.

The satisfiability problem of $\LTL$ under team semantics is defined as follows.
 \decisionproblem{LTL satisfiability w.r.t.\ teams ($\TSAT^\star$) for $\star\in\{\ra,\rs\}$.}{\LTL formula $\phi$.}{Is there a non-empty team~$T$ such that $T\starmodels\phi$} 
The non-emptiness condition is necessary, as otherwise every formula is satisfiable due to the empty team property (see Fig.~\ref{fig:structuralpropertiesoverview}).

We consider the generalisation of the path checking problem for \LTL (denoted by $\LTLPC$), which asks for a given ultimately periodic trace~$t$ and a given formula~$\phi$, whether $t \ltlmodels\phi$ holds. In the team semantics setting, the corresponding question is whether a given finite team comprised of ultimately periodic traces satisfies a given formula. Such a team is given by a team encoding~$\teamup$. To simplify our notation, we will write $\teamup\starmodels \varphi$ instead of $\eval{\teamup} \starmodels \varphi$.

 \decisionproblem{TeamPathChecking ($\mathrm{TPC}^\star$) for $\star\in\{\ra,\rs\}$.}{\LTL formula $\phi$ and a finite team encoding~$\teamup$.}{$\teamup\starmodels\phi$}

Consider the generalised model checking problem where one checks whether the team of traces of a Kripke structure satisfies a given formula. This is the natural generalisation of the model checking problem for classical semantics, denoted by $\LTLMC$, which asks, for a given Kripke structure~$\kripke$ and a given \LTL formula~$\phi$, whether $t \ltlmodels \phi$ for every trace $t$ of $\kripke$.

A Kripke structure~$\kripke = (W, R, \eta, w_I)$ consists of a finite set~$W$ of worlds, a left-total transition relation~$R \subseteq W \times W$, a labeling function~$\eta \colon W \rightarrow \pow{\ap}$, and an initial world~$w_I \in W$. A path~$\pi$ through $\kripke$ is an infinite sequence~$\pi = \pi(0) \pi(1) \pi(2) \cdots \in W^\omega$ such that $\pi(0) = w_I$ and $(\pi(i), \pi(i+1)) \in R$ for every $i \ge 0$. The trace of $\pi$ is defined as $t(\pi) = \eta(\pi(0))\eta(\pi(1))\eta(\pi(2)) \cdots \in (\pow{\ap})^\omega$. The Kripke structure~$\kripke$ induces the team~$T(\kripke) = \set{t(\pi) \mid \pi \text{ is a path through }\kripke}$. 

  \decisionproblem{TeamModelChecking ($\mathrm{TMC}^\star$) for $\star\in\{\ra,\rs\}$.}{\LTL formula $\phi$ and a Kripke structure $\kripke$.}{$T(\kripke)\starmodels\phi$}

\begin{comment}
Finally, we will collect from the literature the complexity of satisfiability, path checking, and model checking for classical $\LTL$. 
The satisfiability and model checking problem of $\LTL$ are well-understood (see, e.g.,~\cite{DemriS02,SistlaC85} and the references in the former paper).

\begin{proposition}[\cite{SistlaC85}]\label{LTLsat}
$\LTL$ satisfiability and $\LTLMC$  are $\PSPACE$-complete w.r.t.\ $\leqpm$-reductions.
\end{proposition}

The exact complexity of $\LTLPC$ is an open problem; the best bounds on the problem are as follows:

\begin{proposition}[\cite{kf09,ms03,Kuhtz10}]
$\LTLPC\in\AC1(\LOGDCFL)$ and is $\NC1$-hard w.r.t.\ $\leqcd$-reductions.
\end{proposition}
\end{comment}

\section{Basic Properties}

We consider several standard properties of team semantics (cf., e.g.~\cite{DKV16}) and verify which of these hold for our two semantics for \LTL. These properties are later used to analyse the complexity of the satisfiability and model checking problems. To simplify our notation, $\starmodels$ denotes $\amodels$ or $\smodels$. See Figure~\ref{fig:structuralpropertiesoverview} for the definitions of the properties and a summary for which semantics the properties hold. 
%
%
%
%We say that a semantics~$\starmodels$ 
%\begin{itemize}
%	\item satisfies the \emph{empty team property}, if every formula is satisfied by the empty team, i.e., if $\emptyset \starmodels \phi$ for every \LTL formula~$\phi$,
%	\item is \emph{downwards closed}, if for every team $T\subseteq(\pow{\ap})^\omega$,
%	$T\starmodels\varphi$ and $T'\subseteq T$  implies $ T'\starmodels\varphi$,
%	\item satisfies \emph{singleton equivalence}, if for every trace~$t \in (\pow{\ap})^\omega$, $\{t\}\starmodels\varphi$ if and only if $t \ltlmodels \phi$,
%	\item is \emph{union closed}, if for all teams~$T,T'\subseteq(\pow{\ap})^\omega$:
%$T\starmodels\varphi$ and $T'\starmodels\varphi $ implies $T\cup T'\starmodels\varphi$,
%	\item is \emph{flat}, if for every team~$T\subseteq(\pow{\ap})^\omega$:  $T\starmodels\varphi$ if and only if $\{t\}\starmodels\varphi$ for all $t\in T$.
%\end{itemize}
%
The positive results follow via simple inductive arguments. For the fact that synchronous semantics is not union closed, consider teams $T = \set{\set{p} \emptyset^\omega}$ and $T' = \set{\emptyset\set{p} \emptyset^\omega}$. Then, we have $T \smodels \F p$ and $T'\smodels \F p$ but $T \cup T' \nsmodels \F p$. Note also that flatness is equivalent of being both downward and union closed.
It turns out that, by Figure~\ref{fig:structuralpropertiesoverview}, $\LTL$ under asynchronous team semantics is essentially classical LTL with a bit of universal quantification: for a team~$T$ and an $\LTL$-formula~$\phi$, we have
$T\amodels \phi$ if and only if $\forall t\in T: t\ltlmodels\phi$.
This however does not mean that $\LTL$ under asynchronous team semantics is not worth of a study; it only means that asynchronous $\LTL$ is essentially classical $\LTL$ if we do not introduce additional atomic formulas that describe properties of teams directly. 
This is a common phenomenon in the team semantics setting. For instance, team semantics of first-order logic has the flatness property, but its extension by so-called \emph{dependence atoms}, is equi-expressive with existential second-order logic \cite{vaananen07}.  
Extensions of $\LTL$  under team semantics are discussed in Section~\ref{sec:extensions}.

At this point, it should not come as a surprise that, due to the flatness property and singleton equivalence, the complexity of satisfiability, path checking, and model checking for $\LTL$ under asynchronous team semantics coincides with those of classical $\LTL$ semantics. 
Firstly, note that an $\LTL$-formula $\phi$ is satisfiable under asynchronous or synchronous team semantics if and only if there is a singleton team that satisfies the formula. 
Secondly, note that to check whether a given team satisfies $\phi$ under asynchronous semantics, it is enough to check whether each trace in the team satisfies $\phi$ under classical $\LTL$; this can be computed by an $\AC0$-circuit using oracle gates for $\LTLPC$. 
Putting these observations together, we obtain the following results from the identical results for $\LTL$ under classical semantics \cite{Kuhtz10,kf09,ms03, SistlaC85}.

The circuit complexity class $\AC{i}$ encompass of polynomial sized circuits of depth $O(\log^i(n))$ and unbounded fan-in; $\NC{i}$ is similarly defined but with bounded fan-in.
A language $A$ is \emph{constant-depth reducible} to a language $B$, in symbols $A\leqcd B$, if there exists a logtime-uniform $\AC0$-circuit family with oracle gates for $B$ that decides membership in $A$.
In this context, \emph{logtime-uniform} means that there exists a deterministic Turing machine that can check the structure of the circuit family $\mathcal C$ in time $O(\log|\mathcal C|)$. For further information on circuit complexity, we refer the reader to the textbook of Vollmer~\cite{vol99}.
Furthermore, $\LOGDCFL$ is the set of languages which are logspace reducible to a deterministic context-free language.

\begin{proposition}\label{prop:TMCa-TSATa-PSPACE}\label{thm:TPCa}$\,$
\begin{enumerate}
	\item $\TMCa$, $\TSATa$, and $\TSATs$ are $\PSPACE$-complete w.r.t.\ $\leqpm$-reductions.
	\item $\TPCa$ is in $\AC1(\LOGDCFL)$ and $\NC1$-hard w.r.t.\ $\leqcd$-reductions.
\end{enumerate}
\end{proposition}

%%%%%
%\begin{comment}
%
%First, we again consider the model checking problem under asynchronous semantics, which boils down to classical \LTL model checking, whose complexity is well-understood (see, e.g., the work of Baier and Katoen~\cite{BaierKatoen08}, for an overview with numerous references to original work). This problem, denoted by $\LTLMC$, asks, for a given Kripke structure~$\kripke$ and a given \LTL formula~$\phi$, whether $t \models \phi$ for every $t \in T(\kripke)$.
%
%
%
%\begin{theorem}\label{thm:TMCa}
%	$\TMCa \equivcd \LTLMC$.
%\end{theorem}
%\begin{proof}
%We have $T(\kripke) \amodels \phi$ if and only if $t \ltlmodels \phi$ for all $t \in T(\kripke)$, due to flatness. Accordingly, the desired reduction is trivial.	
%\end{proof}
%
%
%\begin{theorem}\label{thm:tpca}
%	$\TPCa\equivcd\LTLPC$.
%\end{theorem}
%
%\begin{proof}
%We have $\langle \varphi,\teamup\rangle\in\TPCa$ if and only if $\bigwedge_{(t_0,t_1)\in \teamup}(\langle\varphi,(t_0,t_1)\rangle\in\LTLPC)$. This can be computed by an $\AC0$-circuit using oracle gates for $\LTLPC$.
%
%Now, $\langle \varphi,(t_0,t_1)\rangle\in\LTLPC$ if and only if $\langle\varphi,\set{(t_0,t_1)}\rangle\in\TPCa$ yields the other direction.	
%\end{proof}
%\end{comment}
%%%%%

\section{Classification of Decision Problems Under Synchronous Semantics}

In this section, we examine the computational complexity of path and model checking with respect to the synchronous semantics. 
Our main result settles the complexity of $\mathrm{TPC}^\rs$. 
It turns out that this problem is harder than the asynchronous version.

\begin{figure}
\centering
\begin{tikzpicture}[x=1cm,scale=0.65,y=1.2cm]

%    \begin{scope}[on background layer]
%		\draw[dotted,black] (-1,0) -- (2,0);
%		\draw[dotted,black] (-1,-1) -- (2,-1);
%		\draw[dotted,black] (-1,-2) -- (2,-2);
%	\end{scope}

\node at (0.5,.75) {$U(i)$};

	\foreach \n/\x/\y/\pos/\lab/\col in {1/0/0/180//dotwhite,2/0/-1/180/$\substack{q_i\\\$}$/dot,3/0/-2/180/$\substack{\$}$/dot,4/1/0/0//dot,5/1/-1/0/$\substack{\$}$/dot,6/1/-2/0/$\substack{q_i\\\$\\\#}$/dot}{
		\node[\col,label={\pos:\lab}] (\n) at (\x,\y) {};
	}
	
	\foreach \f/\t in {1/2,2/3,3/4,4/5,5/6}{
		\path[-latex,black] (\f) edge (\t);
	}
	\path[-latex,black] (6) edge[out=45,in=45,looseness=1.5] (1);

	\node at (0,-3) {};
\end{tikzpicture}\hfill
\begin{tikzpicture}[x=1cm,scale=0.65,y=1.2cm]

\node at (.5,1.3) {$E(i)$};

\draw [decorate,decoration={brace,amplitude=5pt,raise=4pt}]
		(-.6,.7) -- (1.8,.7);

\node at (0,.55) {$\scriptstyle T(i,1)$};
\node at (1.25,.55) {$\scriptstyle T(i,0)$};

	\foreach \n/\x/\y/\pos/\lab/\col in {1/0/0/180//dotwhite,2/0/-1/180/$\substack{x_i\\q_i\\\$}$/dot,3/0/-2/180/$\substack{\$\\\#}$/dot,4/1/0/0//dotwhite,5/1/-1/0/$\substack{\$}$/dot,6/1/-2/0/$\substack{x_i, q_i\\\$, \#}$/dot}{
		\node[\col,label={\pos:\lab}] (\n) at (\x,\y) {};
	}
	
	\foreach \f/\t in {1/2,2/3,4/5,5/6}{
		\path[-latex,black] (\f) edge (\t);
	}
	\path[-latex,black] (6) edge[out=45,in=45,looseness=1] (4);
	\path[-latex,black] (3) edge[out=45,in=45,looseness=1] (1);
	
	\node at (0,-3) {};
\end{tikzpicture}\hfill
	\begin{tikzpicture}[x=2.5cm,scale=0.65,y=1.2cm]
		\node[text width=2.5cm,align=left] at (-.15,.75) {\footnotesize if $\ell_{jk}=x_i$,\newline then $L({j,k}):$};
		\node[dotwhite] (1) at (0,0) {};
		\node[dot,label={180:$\substack{x_i\\\$}$}] (2) at (0,-1) {};
		\node[dot,label={180:$\substack{\$\\\#}$}] (3) at (0,-2) {};
		
		\node at (0.65,0) {\footnotesize 1};
		\node at (0.65,-1) {\footnotesize 2};
		\node at (0.65,-2) {\footnotesize 3};

		\node[text width=2.5cm,align=left] at (1.15,.75) {\footnotesize if $\ell_{jk}=\neg x_i$,\newline then $L({j,k}):$};
		\node[dotwhite] (a) at (1,0) {};
		\node[dot,label={0:$\substack{\$}$}] (b) at (1,-1) {};
		\node[dot,label={0:$\substack{x_i\\\$\\\#}$}] (c) at (1,-2) {};

		\foreach \from/\to in {1/2,2/3,a/b,b/c}{
			\path[draw,-latex] (\from) edge (\to);
		}
		
		\path[-latex,black] (c) edge[out=45,in=25,looseness=1] (a);
		\path[-latex,black] (3) edge[out=45,in=25,looseness=1] (1);

		\draw [decorate,decoration={brace,amplitude=5pt,raise=4pt,mirror}]
		(-.1,-2.1) -- (1.1,-2.1) node[xshift=-.65cm,yshift=-.2cm,below,align=center,text width=5cm] {$c_j$ at positions $\{1,2,3\}\setminus\{k\}$};

%		\node[anchor=west,align=left] at (1.8,.2) {add.\ $L(j,k)$ labels:};
%		\node[anchor=west,align=left] at (1.75,-.1) {-- at all vertices $c_{i'}$ for $i'\neq i$};
%		\node[anchor=west,align=left] at (1.75,-.5) {-- at all vertices $x'$ for $x'\neq x$};
%		\node[anchor=west,align=left] at (1.75,-1) {-- at all vertices $\$,\#,q_i$};
%		\node[anchor=west,align=left] at (1.9,-.4) {if $k=1$: $c_j$ at $0,1$};
%		\node[anchor=west,align=left] at (1.9,-.9) {if $k=2$: $c_j$ at $1,2$};
%		\node[anchor=west,align=left] at (1.9,-1.4) {if $k=3$: $c_j$ at $0,2$};

%		\draw [decorate,decoration={brace,amplitude=5pt,raise=4pt}]
%		(4.4,.1) -- (4.4,-1.3) node[midway,text width=2.2cm,xshift=1.6cm] {``global''\newline propositions};

	\end{tikzpicture}

\caption{Traces for the reduction defined in the proof of Lemma~\ref{lem:TPCs_PSPACEhard}.}\label{fig:tracegadgets}
\end{figure}
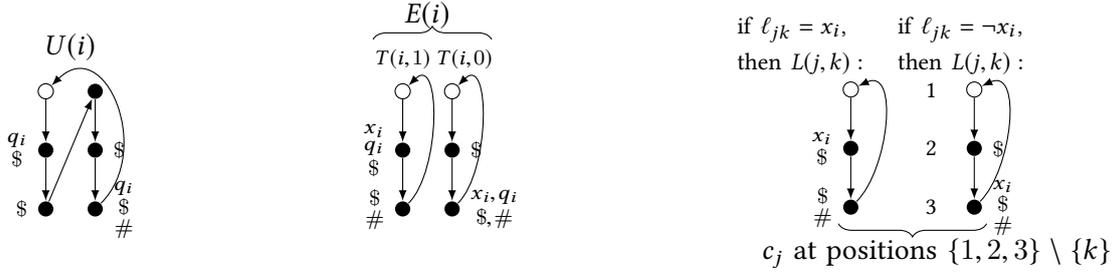

\begin{lemma}\label{lem:TPCs_PSPACEhard}
	$\TPCs$ is $\PSPACE$-hard w.r.t.\ $\leqpm$-reductions.
\end{lemma}
\begin{proof}
Determining whether a given quantified Boolean formula (qBf) is valid ($\qbfval$) is a well-known $\PSPACE$-complete problem \cite{lad77}.
The problem stays $\PSPACE$-complete if the matrix (i.e., the propositional part) of the given qBf is in 3CNF.
To prove the claim of the lemma, we will show that $\qbfval\leqpm \TPCs$.
Given a quantified Boolean formula $\varphi$, we stipulate, w.l.o.g., that $\varphi$ is of the form $\exists x_1\forall x_2\cdots Q x_n\chi$, where $\chi=\bigwedge_{j=1}^m\bigvee_{k=1}^3\ell_{jk}$, $Q\in\{\exists,\forall\}$, and
 $x_1,\dots,x_n$ are exactly the free variables of $\chi$ and pairwise distinct.

In the following we define a reduction which is composed of two functions $f$ and $g$.
Given a qBf $\varphi$, the function $f$ will define an $\LTL$-formula and $g$ will define a team such that $\varphi$ is valid if and only if $g(\varphi)\smodels f(\varphi)$.
Essentially, the team $g(\varphi)$  will contain three kinds of traces, see Figure \ref{fig:tracegadgets}:
(i) traces which are used to mimic universal quantification ($U(i)$ and $E(i)$), (ii) traces that are used to simulate existential quantification ($E(i)$), and (iii) traces used to encode the matrix of $\varphi$ ($L(j,k)$). Moreover the trace $T(i,1)$ ($T(i,0)$, resp.) is used inside the proof to encode an assignment that maps the variable $x_i$ true (false, resp.).
Note that, $U(i), T(i,1), T(i,0), L(j,k)$ are technically singleton sets of traces.
For convenience, we identify them with the traces they contain.

Next we inductively define the reduction function $f$ that maps qBf-formulas to $\LTL$-formulas:
\begin{align*}
	f(\chi) &:= \bigvee_{i=1}^n \F x_i\lor\bigvee_{i=1}^m\F c_i, 
	\intertext{ where $\chi$ is the 3CNF-formula $\bigwedge_{j=1}^m\bigvee_{k=1}^3\ell_{jk}$ with free variables $x_1,\dots, x_n$,}
	f(\exists x_i \psi) &:= (\F q_i)\lor f(\psi),\\
	f(\forall x_i \psi) &:= \bigl(\$ \lor (\lnot q_i\U q_i)\lor \F[\#\land \X f(\psi)]\bigr)\U \#.
\end{align*}

The reduction function $g$ that maps qBf-formulas to teams is defined as follows with respect to the traces in Figure~\ref{fig:tracegadgets}.
\begin{align*}
	g(\chi) &:= \bigcup_{j=1}^mL(j,1)\cup L(j,2)\cup L(j,3), 
\intertext{ where $\chi$ is the 3CNF-formula $\bigwedge_{j=1}^m\bigvee_{k=1}^3\ell_{jk}$ with free variables $x_1,\dots, x_n$ and}
	g(\exists x_i \psi) &:= E(i)\cup g(\psi),\\
	g(\forall x_i \psi) &:= U(i)\cup E(i)\cup g(\psi).
\end{align*}

In Fig.~\ref{fig:tracegadgets}, the first position of each trace is marked with a white circle.
For instance, the trace of $U(i)$ is then encoded via $$(\varepsilon,\emptyset\{q_i,\$\}\{\$\}\emptyset\{\$\}\{q_i,\$,\#\}).$$
The reduction function showing $\qbfval\leqpm \TPCs$ is then $\varphi\mapsto\langle g(\varphi),f(\varphi)\rangle$.
Clearly $f(\varphi)$ and  $g(\varphi)$ can be computed in linear time with respect to $|\varphi|$.

Intuitively, for the existential quantifier case, the formula $(\F q_i)\lor f(\psi)$ allows to continue in $f(\psi)$ with  exactly one of $T(i,1)$ or $T(i,0)$.
If $b\in\{0,1\}$ is a truth value then selecting $T(i,b)$ in the team is the same as setting $x_i$ to $b$.
For the case of $f(\forall x_i\psi)$, the formula $(\lnot q_i\U q_i)\lor \F[\#\land \X f(\psi)]$ with respect to the team $(U(i)\cup E(i))[0,\infty)$ is similar to the existential case choosing $x_i$ to be $1$ whereas for $(U(i)\cup E(i))[3,\infty)$ one selects $x_i$ to be $0$.
The use of the until operator in combination with $\$$ and $\#$ then forces both cases to happen.

Let $\varphi'=Q'x_{n'+1}\cdots Q x_n\chi$, where $Q',Q\in\{\exists,\forall\}$ and let $I$ be an assignment of the variables in $\{x_1,\dots,x_{n'}\}$ for $n'\leq n$. 
Then, let
$$g(I,\varphi') := g(\varphi')\cup\quad\mathclap{\bigcup_{x_i\in\dom(I)}}\quad T(i,I(x_i)).$$
We claim $I\models\varphi'$ if and only if $g(I,\varphi')\smodels f(\varphi')$.

Note that when $\varphi'= \varphi$ it follows that $I=\emptyset$ and that $g(I,\varphi')=g(\varphi)$. 
Accordingly, the lemma follows from the claim of correctness.
The claim is proven by induction on the number of quantifier alternations in $\varphi'$. \iflong
\subparagraph*{Induction basis.}
$\varphi'=\chi$, this implies that $\varphi'$ is quantifier-free and $\dom(I)=\{x_1,\dots,x_n\}$. 
		
``$\Leftarrow$'': 
Let $g(I,\varphi')=T_1\cup T_2$ s.t.\ $T_1\smodels\bigvee_{i=1}^n \F x_i$ and $T_2\smodels\bigvee_{i=1}^n \F c_i$. We assume w.l.o.g.\ $T_1$ and $T_2$ to be disjoint, which is possible due to downwards closure. 
We then have that $T_2\subseteq\{\,L(j,k)\mid 1\leq j\leq m, 1\leq k \leq 3\,\}$ and
$T_1= (\{\,L(j,k)\mid 1\leq j\leq m, 1\leq k \leq 3\,\}\setminus T_2)\cup\{\,T(i,I(x_i))\mid 1\leq i\leq n\,\}$, for, $1\leq i,i'\leq m$, $c_i$ does not appear positively in the trace $T(i',I(x_{i'}))$.
Due to construction of the traces, $L(j,k)\in T_2$ can only satisfy the subformula $\F c_{j'}$ for $j' = j$.
Moreover, note that there exists no $s\in\N$ such that $L(j,k)(s)\ni c_j$ for all $1\leq k\leq 3$; hence $\{L(j,1),L(j,2),L(j,3)\}$ falsifies $\F c_j$.
These two combined imply that $T_2\not\supseteq\{L(j,1),L(j,2),L(j,3)\}$, for each $1\leq j\leq m$.
%\begin{align*}
%	T_2\not\supseteq\{L(j,1),L(j,2),L(j,3)\}\text{ for each }1\leq j\leq m.%\tag{$\heartsuit$}
%\end{align*}
However, for each $1\leq j\leq m$, any two of $L(j,k)$, $1\leq k\leq 3$, can belong to $T_2$ and hence exactly one belongs to $T_1$.

Now $T_1=T_1^1\cup\cdots\cup T_1^n$ such that $T_1^i\smodels \F x_i$.
Note that
$\F x_i$ can be satisfied by $T(i',I(x_{i'}))$ only for $i' = i$. Since $T_1\supseteq\{\,T(i,I(x_i))\mid 1\leq i\leq n\,\}$, it follows that $T(i,I(x_{i})) \in T_1^i$, for each $1\leq i \leq n$.
Note also that, if $L(j,k)\in T_1$ it has to be in $T_1^i$ where $x_i$ is the variable of $\ell_{j,k}$.
By construction of the traces, if $T(i,1)\in T_1^i$ we have $T_1^i(1)\smodels x_i$ and if $T(i,0)\in T_1^i$ then $T_1^i(2)\smodels x_i$. Thus, by construction of the traces $L(j,k)$, if $L(j,k)\in T_1$ then $I\models\ell_{j,k}$.
Since, for each $1\leq j\leq m$, there is a $1\leq k\leq 3$ such that $L(j,k)\in T_1$ it follows that $I\models\varphi'$.

``$\Rightarrow$'': 
Now assume that $I\models\varphi'$.
As a result, pick for each $1\leq j\leq m$ a \emph{single} $1\leq k\leq 3$ such that $I\models\ell_{jk}$.
Denote this sequence of choices by $k_1,\dots,k_m$.
Choose $g(I,\varphi')=T_1\cup T_2$ as follows:
\begin{align*}
	T_1 &:= \{\,L(j,k_j)\mid  1\leq j\leq m \}\cup\{T(i,I(x_i)\mid 1\leq i\leq n\,\}\\
	T_2 &:= \{\,L(j,1),L(j,2),L(j,3)\mid 1\leq j\leq m\,\}\setminus T_1
\end{align*}
Then $T_2\smodels \bigvee_{j=1}^m \F c_j$, for exactly two traces per clause are in $T_2$, and we can divide $T_2=T_2^1\cup\cdots\cup T_2^m$ where
$$
T_2^j := \{\,L(j,k),L(j,k')\mid k,k'\in\{1,2,3\}\setminus\{k_j\}\,\},
$$
and, by construction of the traces, $T_2^j\smodels \F c_j$, for all $1\leq j\leq m$.
% as follows.
%\begin{itemize}
%	\item If $\{L(j,1),L(j,2)\}=T_2$ then $c_j\in L(j,k)(1)$ for $k=1,2$.
%	\item If $\{L(j,1),L(j,3)\}=T_2$ then $c_j\in L(j,k)(0)$ for $k=1,3$.
%	\item If $\{L(j,2),L(j,3)\}=T_2$ then $c_j\in L(j,k)(2)$ for $k=2,3$.
%\end{itemize}
%
Further, note that $T_1=T_1^1\cup\cdots \cup T_1^n$, where
$$
T_1^i \dfn \{\,L(j,k_j)\mid 1\leq j\leq m, I(x_i)\models\ell_{jk}\,\}\cup\{T(i,I(x_i))\}.
$$
There are two possibilities:
\begin{itemize}
	\item $I(x_i)=1$: then $x_i\in(L(j,k_j)(1)\cap T(i,I(x_i))(1))$.
	\item $I(x_i)=0$: then $x_i\in(L(j,k_j)(2)\cap T(i,I(x_i))(2))$. 
\end{itemize}
In both cases, $T_1^i\smodels \F x_i$, and thus $T_1\smodels\bigvee_{i=1}^n \F x_i$.
Hence it follows that $g(I,\varphi')\smodels f(\varphi')$ and the induction basis is proven.
\subparagraph*{Induction Step.} 
		``Case $\varphi'=\exists x_i\psi$.'' 
		We show that $I\models \exists x_i\psi$ if and only if $g(I,\exists x_i\psi)\smodels f(\exists x_i\psi)$. 
		
		First note that $g(I,\exists x_i\psi)\smodels f(\exists x_i\psi)$ iff $E(i)\cup g(\psi)\smodels (\F q_i)\lor f(\psi)$, by the definitions of $f$ and $g$.
		Clearly, $E(i)\nsmodels \F q_i$, but both $T(i,1)\smodels\F q_i$ and $T(i,0)\smodels\F q_i$.
		Observe that $E(i)=\{T(i,1),T(i,0)\}$ and $q_i$ does not appear positively anywhere in $g(\psi)$.
		%It follows that, there exists $T_1\subseteq E(i)$ such that $T_1\models \F q_i$ and $(E(i)\setminus T_1)\cup g(\psi)\models f(\psi)$.
	%
		Accordingly, and by downwards closure, $E(i)\cup g(\psi)\smodels (\F q_i)\lor f(\psi)$ if and only if 
		\begin{equation}
			\exists b\in\{0,1\}:\; T(i,1-b)\smodels \F q_i \text{ and }
			(E(i)\cup g(\psi)) \setminus T(i,1-b) \smodels f(\psi).\label{eq:IS-exists-1}
		\end{equation}
		Since $(E(i)\cup g(\psi)) \setminus T(i,1-b) = T(i,b)\cup g(\psi) = g(I[x_i \mapsto b],\psi)$, Equation~\eqref{eq:IS-exists-1} holds if and only if $g(I[x_i \mapsto b],\psi) \smodels f(\psi)$, for some bit $b\in\{0,1\}$. By the induction hypothesis, the latter holds if and only if there exists a bit $b\in\{0,1\}$ s.t.\ $I[x_i\mapsto b]\models\psi$. Finally by the semantics of $\exists$  this holds if and only if $I\models\exists x_i\psi$.
		
``Case $\varphi'=\forall x_i\psi$.'' 
		We need to show that $I\models \forall x_i\psi$ if and only if $g(I,\forall x_i\psi)\smodels f(\forall x_i\psi)$. 
		
		First note that, by the definitions of $f$ and $g$, we have 
		$$g(I,\forall x_i\psi)\smodels f(\forall x_i\psi)$$ if and only if
\begin{equation}\label{unicase1}
U(i)\cup E(i)\cup g(\psi) \smodels \bigl(\$\! \lor\! (\lnot q_i\U q_i)\!\lor\! \F[\#\!\land\! \X f(\psi)]\bigr)\U \#.
\end{equation}
In the following, we will show that \eqref{unicase1} is true if and only if $T(i,b) \cup  g(\psi) \smodels f(\psi)$ for all $b\in\{0,1\}$. 
From this the correctness follows analogously as in the case for the existential quantifier.

Notice first that each trace in $U(i)\cup E(i)\cup g(\psi)$ is periodic with period length either $3$ or $6$, and exactly the last element of each period is marked by the symbol~$\#$. 
Consequently, it is easy to see that \eqref{unicase1} is true if and only if
\begin{equation}\label{unicase2}
(U(i)\!\cup\! E(i)\!\cup\! g(\psi))[j,\infty) \smodels \$\! \lor\! (\lnot q_i\U q_i)\!\lor\! \F[\#\!\land\! \X f(\psi)],
\end{equation}	
for each $j\in\{0,1,2,3,4\}$.
Note that
$$(U(i)\cup E(i)\cup g(\psi))[j,\infty) \smodels \$,$$ for each $j\in\{1,2,4\}$, whereas no non-empty subteam of $(U(i)\cup E(i)\cup g(\psi))[j,\infty)$, $j\in\{0,3\}$ satisfies $\$$. Accordingly, \eqref{unicase2} is true if and only if
\begin{equation}\label{unicase3}
(U(i)\cup E(i)\cup g(\psi))[j,\infty) \smodels (\lnot q_i\U q_i)\!\lor\! \F[\#\land \X f(\psi)],
\end{equation}
for both $j\in\{0,3\}$.
Note that, by construction, $q_i$ does not occur positively in $g(\psi)$. 
As a result, $X\cap g(\psi)[j,\infty)= \emptyset$, $j\in\{0,3\}$, for all teams $X$ s.t. $X\smodels \lnot q_i\U q_i$.  
Also, none of the symbols $x_{i'}$, $c_{i'}$, $q_{i''}$, for $i'$, $i'' \in \nats$ with $i''\neq i$, occurs positively in $U(i)$. 
On that account,  $X\cap U(i)[j,\infty)= \emptyset$, $j\in\{0,3\}$, for all $X$ s.t.\ $X\smodels \F[\#\land \X f(\psi)]$, for eventually each trace in $X$ will end up in a team that satisfies one of the formulas of the form $\F x_{i'}$, $\F c_{i'}$, or $\F q_{i''}$ (see the inductive definition of $f$).
Moreover, it is easy to check that $(T(i,1)\cup U(i))[0,\infty)\smodels \lnot q_i\U q_i$,  $(T(i,0)\cup U(i))[0,\infty)\nsmodels \lnot q_i\U q_i$, $(T(i,0)\cup U(i))[3,\infty)\smodels \lnot q_i\U q_i$, and $(T(i,1)\cup U(i))[3,\infty)\nsmodels \lnot q_i\U q_i$. 
From these, together with downwards closure, it follows that \eqref{unicase3} is true if and only if for $b_0=1 $ and $b_3=0$
\begin{equation}\label{unicase4}
(U(i)\cup T(i,b_j))[j,\infty) \smodels \lnot q_i\U q_i, \text{ for all $j\in\{0,3\}$}
\end{equation}	
and
\begin{equation}\label{unicase5}
(T(i,1-b_j) \cup g(\psi))[j,\infty) \smodels \F[\#\land \X f(\psi)], 
\end{equation}
for both $j\in\{0,3\}$.
In fact, as \eqref{unicase4} always is the case, \eqref{unicase3} is equivalent with \eqref{unicase5}. 
By construction, \eqref{unicase5} is true if and only if $(T(i,b) \cup g(\psi))[6,\infty) \smodels f(\psi)$, for both $b\in\{0,1\}$. 
Now, since $$(T(i,b) \cup g(\psi))[6,\infty)=T(i,b) \cup g(\psi)$$ the claim applies.
\else
\infull
\fi
\end{proof}

\iflong
Now we turn our attention to proving a matching upper bound. 
To this end, we need to introduce some notation to manipulate team encodings. 
Given a pair~$(t_0,t_1)$ of traces~$t_0 = t_0(0) \cdots t_0(n)$ and $t_1 = t_1(0)\cdots t_1(n')$, we define $(t_0,t_1)[1,\infty)$ to be $(t_0(1) \cdots t_0(n), t_1)$ if $t_0 \neq \epsilon$, and to be $(\epsilon, t_1(1) \cdots t_1(n')t_1(0))$ if $t_0 = \epsilon$. 
Furthermore, we inductively define $(t_0,t_1)[i,\infty)$ to be $(t_0,t_1)$ if $i = 0$, and to be $((t_0,t_1)[1,\infty))[i-1,\infty)$ if $i>0$. 
Then, $$\eval{(t_0,t_1)[i,\infty)} = (\eval{(t_0,t_1)})[i,\infty),$$ that is, we have implemented the prefix-removal operation on the finite representation. Furthermore, we lift this operation to team encodings~$\teamup$ by defining $\teamup[i,\infty) = \set{ (t_0,t_1)[i,\infty) \mid (t_0,t_1) \in \teamup }$. As a result, we have $\eval{\teamup[i,\infty)} = (\eval{\teamup})[i,\infty)$.
 
Given a finite team encoding $\teamup $, let $$\prfx(\teamup) = \max \set{\size{t_0} \mid (t_0, t_1) \in \teamup}$$ and let $\lcm(\teamup)$ be the least common multiple of $\set{\size{t_1} \mid (t_0, t_1) \in \teamup}$. Then, $\teamup[i,\infty) = \teamup[i+\lcm(\teamup),\infty)$ for every $i \ge \prfx(\teamup)$. The next remark is now straightforward.

Furthermore, observe that if $\teamup$ is a finite team encoding and let $i \ge \prfx(\teamup)$. Then, $\teamup[i,\infty)$ and $\teamup[i+\lcm(\teamup),\infty)$ satisfy exactly the same \LTL formulas under synchronous team semantics. 

In particular, we obtain the following consequences for temporal operators (for finite $\teamup$):

\begin{tabbing}
    $\teamup \models \psi \R \phi$ \= iff \= Rechts \kill
	$\teamup \models \F\varphi$ \> iff \> $\exists k \leq \prfx(\teamup)+\lcm(\teamup): \teamup[k, \infty) \models \varphi$.\\
	$\teamup \models \G\varphi$ \> iff \> $\forall k \leq \prfx(\teamup)+\lcm(\teamup): \teamup[k, \infty)\models \varphi$\\
	$\teamup \models \psi \U \phi$ \> iff \> $\exists k \leq \prfx(\teamup)+\lcm(\teamup): \teamup[k, \infty) \models \phi$ and 
	$\forall k' < k: \teamup[k', \infty) \models \phi$\\
	$\teamup \models \psi \R \phi$ \> iff \> $\forall k \leq \prfx(\teamup)+\lcm(\teamup): \teamup[k, \infty) \models \phi$ or
	$\exists k' < k: \teamup[k', \infty) \models \phi$
\end{tabbing}

Accordingly, we can restrict the range of the temporal operators when model checking a finite team encoding. This implies that a straightforward recursive algorithm implementing the synchronous semantics solves $\LTLPC^\rs$.

%\todo[inline]{Jonni: Maybe change to a NPSPACE algorithm for clarity. It might also be better to just describe the algortithm. In principle, we just need to guess splits, and natural numbers at most exponential in size, and then check.}
%\todo[inline]{Arne: This could also be option to create more space if needed}
\begin{lemma}\label{lem:TPCs_in_PSPACE}
	$\TPCs$ is in $\PSPACE$.
\end{lemma}
\begin{proof}
Consider Alg.~\ref{algo_pathchecking} where $\lor$ and $\bigvee$ denote classical disjunction, not splitjunction, when combining results from recursive calls.

\SetKw{proce}{Procedure}
\SetKwFunction{pcheck}{chk}
\SetKwFunction{nitram}{nitram}

\LinesNumbered
\begin{algorithm}\caption{Algorithm for $\TPCs$.}\label{algo_pathchecking}
\small
\proce \pcheck{Team encoding $\teamup$, formula $\phi$}\;
	\lIf{$\varphi=p$}{\Return $ \bigwedge_{(t_0,t_1)\in\teamup}\,\, p \in t_0t_1(0)$}
	\lIf{$\varphi=\neg p$}{\Return $ \bigwedge_{(t_0,t_1)\in\teamup}\,\, p \notin t_0t_1(0)$}
	\lIf{$\varphi=\psi\land\psi'$}{\Return \pcheck{$\teamup,\psi$} $\land$ \pcheck{$\teamup,\psi'$}}
	\lIf{$\varphi=\psi\lor\psi'$}{\Return $\bigvee_{\teamup' \subseteq \teamup}$\,\,\pcheck{$\teamup',\psi$} $\land$ \pcheck{$\teamup\setminus \teamup',\psi'$}}
	\lIf{$\varphi=\X\psi$}{\Return \pcheck{$\teamup[1,\infty),\psi$}}
	\lIf{$\varphi=\F\psi$}{\Return $\bigvee_{k \le \prfx(\teamup) + \lcm(\teamup)}$\,\,\pcheck{$\teamup[k,\infty),\psi$}}
	\lIf{$\varphi=\G\psi$}{\Return $\bigwedge_{k \le \prfx(\teamup) + \lcm(\teamup)}$\,\,\pcheck{$\teamup[k,\infty),\psi$}}
	\lIf{$\varphi=\psi\U\psi'$}{\Return $\left(\right.$ $\bigvee_{k \le \prfx(\teamup) + \lcm(\teamup)}$\,\,\pcheck{$\teamup[k,\infty),\psi'$} $\land$ $\bigwedge_{k' < k}$\,\,\pcheck{$\teamup[k',\infty),\psi$}$\left.\right)$}
	\lIf{$\varphi=\psi\R\psi'$}{\Return $\left(\right.$ $\bigwedge_{k \le \prfx(\teamup) + \lcm(\teamup)}$\,\,\pcheck{$\teamup[k,\infty),\psi'$} $\lor$ $\bigvee_{k' < k}$\,\,\pcheck{$\teamup[k',\infty),\psi$}$\left.\right)$}

\end{algorithm}

The algorithm is an implementation of the synchronous team semantics for \LTL with slight restrictions to obtain the desired complexity. In line~5, we only consider strict splits, i.e., the team is split into two disjoint parts. This is sufficient due to downwards closure. Furthermore, the scope of the temporal operators in lines~7 to 10 is restricted to the interval $[0, \prfx(\teamup) + \lcm(\teamup)]$. This is sufficient due to the above observations. 

It remains to analyse the algorithm's space complexity. Its recursion depth is bounded by the size of the formula. Further, in each recursive call, a team encoding has to be stored. Additionally, in lines 5 and 7 to 10, a disjunction or conjunction of exponential arity has to be evaluated. In each case, this only requires linear space in the input to make the recursive calls and to aggregate the return value. 
Thus, Algorithm~\ref{algo_pathchecking} is implementable in polynomial space. 
\end{proof}

Combining Lemma~\ref{lem:TPCs_PSPACEhard} and \ref{lem:TPCs_in_PSPACE} settles the complexity of $\TPCs$.
\else
The matching upper bound follows via a $\PSPACE$ algorithm implementing the semantics in straightforward way. \infull
\fi
\begin{theorem}\label{thm:tpcs}
	$\TPCs$ is $\PSPACE$-complete w.r.t.\ $\leqpm$-reductions.
\end{theorem}

\iflong
\subsection{Model Checking}
\label{subsec:modelchecking}
\fi

The next theorem deals with model checking of the split\-junc\-tion-free fragment of \LTL under synchronous team semantics.
\begin{theorem}\label{thm:tmcs-splitfree}
$\TMCs$ restricted to splitjunction-free formulas is in $\PSPACE$.
\end{theorem}
\begin{proof}
Fix $\kripke = (W, R, \eta, w_I)$ and a splitjunction-free formula~$\phi$. We define $S_0 = \set{w_I}$ and $S_{i+1} = \set{w' \in W \mid (w,w') \in R \text{ for some }w \in S_i}$ for all $i \ge 0$. By the pigeonhole principle, this sequence is ultimately periodic with a characteristic~$(s,p)$ with $s+p \le 2^{\size{W}}$.\footnote{The characteristic of an encoding~$(t_0, t_1)$ of an ultimately periodic trace~$t_0 t_1t_1 t_1 \cdots$ is the pair~$(\size{t_0}, \size{t_1})$. Slightly abusively, we say that $(\size{t_0}, \size{t_1})$ is the characteristic of $t_0 t_1t_1 t_1 \cdots$, although this is not unique.} Next, we define a trace~$t$ over $\ap \cup \set{\overline{p} \mid p \in \ap}$ via
\[
t(i) =\, \set{p \in \ap \mid p \in \eta(w) \text{ for all } w \in S_i} \cup
\set{\overline{p} \mid p \notin \eta(w) \text{ for all } w \in S_i}	
\]
that reflects the team semantics of (negated) atomic formulas, which have to hold in every element of the team. 

An induction over the construction of $\phi$ shows that $T(\kripke) \smodels \phi$ if and only if $t \ltlmodels \overline{\phi}$, where $\overline{\phi}$ is obtained from $\phi$ by replacing each negated atomic proposition~$\neg p$ by $\overline{p}$. To conclude the proof, we show that $t \ltlmodels \overline{\phi}$ can be checked in non-deterministic polynomial space, exploiting the fact that $t$ is ultimately periodic and of the same characteristic as $S_0S_1S_2 \cdots$. However, as $s+p$ might be exponential, we cannot just construct a finite representation of $t$ of characteristic~$(s, p)$ and then check satisfaction in polynomial space.

Instead, we present an on-the-fly approach which is inspired by similar algorithms in the literature. It is based on two properties:
\begin{enumerate}
	\item Every $S_i$ can be represented in polynomial space, and from $S_i$ one can compute $S_{i+1}$ in polynomial time. 
	\item For every \LTL formula~${\overline{\phi}}$, there is an equivalent non-deter\-mi\-nis\-tic Büchi automaton~$\aut_{\overline{\phi}}$ of exponential size (see, e.g., \cite{BaierKatoen08} for a formal definition of Büchi automata and for the construction of $\aut_{\overline{\phi}}$). States of $\aut_{\overline{\phi}}$ can be represented in polynomial space and given two states, one can check in polynomial time, whether one is a successor of the other. 
\end{enumerate}
These properties allow us to construct both $t$ and a run of $\aut_{\overline{\phi}}$ on $t$ on the fly.
\iflong
In detail, the algorithm works as follows. It guesses a set~$S^* \subseteq W$ and a state~$q^*$ of $\aut_{\overline{\phi}}$ and checks whether there are $i < j$ satisfying the following properties:
\begin{itemize}
	\item $S^* = S_i = S_j$,
	\item $q^*$ is reachable from the initial state of $\aut_{\overline{\phi}}$ by some run on the prefix~$t(0) \cdots t(i)$, and
	\item $q^*$ is reachable from $q^*$ by some run on the infix~$t(i+1) \cdots t(j)$. This run has to visit at least one accepting state. 
\end{itemize}
By an application of the pigeonhole principle, we can assume w.l.o.g.\ that $j$ is at most exponential in $\size{W}$ and in $\size{\phi}$.

Let us argue that these properties can be checked in non-deter\-mi\-nis\-tic polynomial space. Given some guessed $S^*$, we can check the existence of $i<j$ as required by computing the sequence~$S_0 S_1 S_2\cdots$ on-the-fly, i.e., by just keeping the current set in memory, comparing it to $S^*$, then computing its successor, and then discarding the current set. While checking these reachability properties, the algorithm also guesses corresponding runs as required in the second and third property. As argued above, both tasks can be implemented in non-deterministic space. To ensure termination, we stop this search when the exponential upper bound on $j$ is reached. This is possible using a counter with polynomially many bits and does not compromise completeness, as argued above. 

It remains to argue that the algorithm is correct. First, assume $t \ltlmodels \overline{\phi}$, which implies that $\aut_{\overline{\phi}}$ has an accepting run on $t$. Recall that $t$ is ultimately periodic with characteristic~$(s,p)$ such that $s+p \le 2^{\size{W}}$ and that $\aut_{\overline{\phi}}$ is of exponential size. As a result, a pumping argument yields~$i <j$ with the desired properties.

Secondly, assume the algorithm finds~$i<j$ with the desired properties. Then, the run to $q$ and the one from $q$ to $q$ can be turned into an accepting run of $\aut_{\overline{\phi}}$ on $t$. That being so, $t \ltlmodels{\overline{\phi}}$.\else\infull\fi
\end{proof}
\iflong
Note that our algorithm is even able to deal with arbitrary negations, as long as we disallow splitjunctions.
\fi

The complexity of general model checking problem is left open. It is trivially $\PSPACE$-hard, due to Theorem~\ref{thm:tpcs} and the fact that finite teams of ultimately periodic traces can be represented by Kripke structures.
However, the problem is potentially much harder, as one has to deal with infinitely many splits of possibly uncountable teams with non-periodic traces, if a split occurs under the scope of a $\G$-operator. Currently, we are working on interesting language-theoretic problems one encounters when trying to generalise our algorithms for the general path checking problem and for the splitjunction-free model checking problem, e.g., how complex can an \LTL-definable split be, if the team to be split is one induced by a Kripke structure.

\newcommand{\agreeson}[3]{#2\overset{#1}{\Leftrightarrow}#3}
\section{Extensions}\label{sec:extensions}
In this section we take a brief look into extensions of our logics by dependence atoms and contradictory negation.
Contradictory negation combined with team semantics allows for powerful constructions.
For instance, the complexity of model checking for propositional logic jumps from $\NC{1}$ to $\PSPACE$ \cite{muellerDiss}, whereas the complexity of validity and satisfiability jumps all the way to alternating exponential time with polynomially many alternations ($\ATIME(\exp,\pol)$) \cite{Hannula:2018:CPL:3176362.3157054}.

Formally, we define that $T\starmodels\!\!\sim\!\varphi$ if $T\nstarmodels\varphi$.
Note that the negation $\sim$ is not equivalent to the negation $\neg$ of atomic propositions defined earlier, i.e., $\sim\!\! p$ and $\neg p$ are not equivalent.
In the following, problems of the form $\TPCa(\sim)$, etc., refer to \LTL-formulas with negation $\sim$.

Also, we are interested in atoms expressible in first-order (FO) logic over the atomic propositions; 
the most widely studied ones are dependence, independence, and inclusion atoms \cite{DKV16}.
The notion of generalised atoms in the setting of first-order team semantics was introduced by Kuusisto \cite{kuusisto15}.
It turns out that \iflong Algorithm~\ref{algo_pathchecking}\else the algorithm for $\TPCs$\ \fi%
is very robust to such strengthenings of the logic under consideration. 

We consider FO-formulas over the signature $(A_p)_{p \in \ap}$, where each $A_p$ is a unary predicate. Furthermore, we interpret a team~$T$ as a relational structure~$\mathfrak{A}(T)$ over the same signature with universe~$T$ such that $t \in T$ is in $ A^{\mathfrak A}_{p}$  if and only if $p \in t(0)$. The formulas then express properties of the atomic propositions holding in the initial positions of traces in $T$.
An FO-formula~$\varphi$ \emph{FO-defines} the atomic formula~$D$ with $T\starmodels D \Longleftrightarrow \mathfrak A(T)\models\varphi$. 
In this case, $D$ is also called an \emph{FO-definable} \emph{generalised atom}.

For instance, the dependence atom $\dep(x;y)$ is FO-definable by 
$
\forall t\forall t'
((A_{x}(t)\leftrightarrow A_{x}(t'))\to\bigl(A_y(t)\leftrightarrow A_y(t')))$, for $x,y \in \ap$.
We call an $\LTL$-formula extended by a generalised atom~$D$ an $\LTL(D)$-formula.
Similarly, we lift this notion to \emph{sets of generalised atoms} as well as to the corresponding decision problems, i.e., $\TPCs(D)$ is the path checking problem over synchronous semantics with $\LTL$ formulas which may use the generalised atom $D$.

The result of Theorem \ref{thm:tpcs} can be extended to facilitate also the contradictory negation and first-order definable generalised atoms.	
\iflong
\begin{theorem}
	Let $\mathcal{D}$ be a finite set of first-order definable generalised atoms.
	Then $\TPCs(\mathcal D)$ is $\PSPACE$-complete w.r.t.\ $\leqpm$-reductions.
\end{theorem}
\begin{proof}
The lower bound applies from Theorem~\ref{thm:tpcs}. 
For the upper bound, we extend the algorithm stated in the proof of Lemma~\ref{lem:TPCs_in_PSPACE} for the cases of FO-definable atoms.
Whenever such an atom $D$ appears in the computation of the algorithm, we need to solve an FO model checking problem.
As FO model checking is solvable in logarithmic space \cite{Immerman1998} the theorem follows.
\end{proof}

Alg.~\ref{algo_pathchecking} for $\TPCs$ can be straightforwardly extended to deal with contradictory negations without a price in terms of complexity. %as we will see in the proof of the next theorem.

 \begin{theorem}\label{thm:tpcs-negation-pspacecomplete}
  	$\TPCs(\sim)$ is $\PSPACE$-complete w.r.t.\ $\leqpm$-reductions.
  \end{theorem}
\begin{proof}
The lower bound follows from Theorem~\ref{thm:tpcs}, while the upper bound is obtained by adding the following line to the recursive algorithm from the proof of Lemma~\ref{lem:TPCs_in_PSPACE} (where $\neg$ denotes classical negation): 
  	$ \text{	\textbf{if } $\phi = \sim\!\phi'$ \textbf{then return} $\neg$\texttt{chk}$(\teamup,\phi')$ }$
\end{proof}
\else
\begin{theorem}\label{thm:tpcs-fo-npcomplete}
	Let $\mathcal{D}$ be a finite set of first-order definable generalised atoms.
	Then $\TPCs(\mathcal D)$ and $\TPCs(\sim)$ are $\PSPACE$-complete w.r.t.\ $\leqpm$-reductions.
\end{theorem}
\fi

The next proposition translates a result from Hannula~et~al.~\cite{Hannula:2018:CPL:3176362.3157054} to our setting. 
They show completeness for $\ATIME(\exp,\pol)$ for the satisfiability problem of propositional team logic with negation.
This logic coincides with LTL-formulas without temporal operators under team semantics. 
\begin{proposition}[\cite{Hannula:2018:CPL:3176362.3157054}]\label{prop:atime(exp,pol)}
	$\TSATa(\sim)$ and $\TSATs(\sim)$ for formulas without temporal operators are complete for $\ATIME(\exp,\pol)$ w.r.t.\ $\leqpm$-reductions.
\end{proposition}

%The result from the previous proposition will be utilised in the proof of the next theorem.

\begin{theorem}\label{thm:tmca-tmcs-negation}
$\TMCa(\sim)$ and $\TMCs(\sim)$ are hard for $\ATIME(\exp,\pol)$  w.r.t.\ $\leqpm$-reductions. 
\end{theorem}
\begin{proof}
We will state a reduction from the satisfiability problem of propositional team logic with negation $\sim$ (short $\PL(\sim)$). The stated hardness then follows from Proposition~\ref{prop:atime(exp,pol)}.

\begin{figure}
\centering
	\begin{tikzpicture}[x=.5cm,y=.75cm]
		\node at (-1.5,0) {$\mathcal K_P$:};
		\node[dw] (root) at (0,0) {\tiny$r$};
		\node[d,label={180:$\substack{p_1}$}] (a1) at (-1,-1) {\tiny$a_1$};
		\node[d,label={0:$\substack{\overline{p_1}}$}] (b1) at (1,-1) {\tiny$b_1$};
		\node[d,label={180:$\substack{p_2}$}] (a2) at (-1,-2) {\tiny$a_2$};
		\node[d,label={0:$\substack{\overline{p_2}}$}] (b2) at (1,-2) {\tiny$b_2$};
		\node[d,label={180:$\substack{p_n}$}] (an) at (-1,-3) {\tiny$a_n$};
		\node[d,label={0:$\substack{\overline{p_n}}$}] (bn) at (1,-3) {\tiny$b_n$};
		
		\foreach \f/\t in {root/a1,root/b1,a1/b2,a1/a2,b1/b2,b1/a2}{
			\path[-stealth',draw] (\f) edge (\t);
		}
		
		\path[-stealth',draw,dotted] (a2) edge (bn);
		\path[-stealth',draw,dotted] (a2) edge (an);
		\path[-stealth',draw,dotted] (b2) edge (an);
		\path[-stealth',draw,dotted] (b2) edge (bn);
		
		\path[-, draw] (an) edge[loop below,->, >=stealth'] (an);
		\path[-, draw] (bn) edge[loop below,->, >=stealth'] (bn);

	\end{tikzpicture}
\caption{Kripke structure for the proof of Theorem~\ref{thm:tmca-tmcs-negation}.}\label{fig:kripke-tmca-tmcs-negation}
\end{figure}
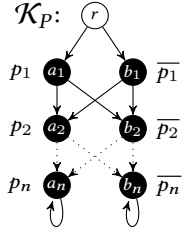
For $P=\{p_1,\dots, p_n\}$, consider the tra\-ces starting from the root $r$ of the Kripke structure $\mathcal{K}_P$ depicted in Figure~\ref{fig:kripke-tmca-tmcs-negation} using proposition symbols $p_1,\dots, p_n,\overline{p_1},\dots, \overline{p_n}$.
Each trace in the model corresponds to a propositional assignment on $P$.
%The reduction from $\PL(\sim)$ to $\LTL(\sim)$ then replaces in $\varphi$ each variable $x_i$ by $\F p_i$ and each negated variable $\lnot x_i$ by $\F\overline{p_i}$.
For $\varphi \in \PL{(\sim)}$, let $\varphi^*$ denote the $\LTL(\sim)$-formula obtained by simultaneously replacing each (non-ne\-gated) variable $p_i$ by $\F p_i$ and each negated variable $\lnot p_i$ by $\F\overline{p_i}$. 
Let $P$ denote the set of variables that occur in $\varphi$. 
Define $\top:= (p \lor \neg p)$ and $\bot:=p \land \neg p$, then $T(\mathcal{K}_P)\starmodels \bigl(\top \lor ((\sim\!\!\bot) \land \varphi^*)\bigr)$ \emph{if and only if} $T' \starmodels \varphi^*$ for some non-empty $T'\subseteq T(\mathcal{K}_P)$. It is easy to check that $T'\starmodels\varphi^*$ \emph{if and only if} the propositional team  corresponding to $T'$ satisfies $\varphi$ and thus the above holds if and only if $\varphi$ is satisfiable.
%
%
%
%Note that we cannot do a split where the right side is empty because then $\sim\bot$ would be false.
%
% Jonni: The following does not hold.
%Consequently, this nonempty split into two subteams $T_1, T_2$ such that $T_1\starmodels\top$ and $T_2\starmodels(\sim \bot \land \varphi^*)$ has the form $|T_2|=1$ containing the satisfying assignment for $\varphi$.
%Note that, the crux why this construction works for both semantics really is the fact~$\size{T_2}=1$.
\end{proof}

%\todo[inline]{Jonni: Can we do something similar in order to encode $SAT(ML(\sim))$? Consider complete graph with propositional assignments as points. Now taking $T(M,root)$ of that graph produces a team that has all possible paths. Now each point can be replaced with a ladder, like in the previous proof. Probably this could be encoded also using exactly one ladder with a new point as a sink, and then connect the sink to the root (here also some proposition in needed to the sink in order facilitate the encoding). Now $SAT(ML(\sim))$ is decidable, but non-elementary.}

%\subsection{Dependence Atoms}\label{sec:dep}
In the following, we define the semantics for dependence atoms.
For Teams~$T \subseteq (\pow{\ap})^\omega$ we define $T\starmodels\!\dep(p_1,\dots,p_n;q_1,\dots,q_m)$ if 
\[
	\forall t,t'\in T:\; (\agreeson{p_1}{t(0)}{t'(0)},\dots,\agreeson{p_n}{t(0)}{t'(0)}) \text{ implies } (\agreeson{q_1}{t(0)}{t'(0)},\dots,\agreeson{q_m}{t(0)}{t'(0)}),
\]
where $\agreeson{p}{t(i)}{t(j)}$ means the sets $t(i)$ and $t(j)$ agree on proposition $p$, i.e., both contain $p$ or not.
Observe that the formula $\dep(;p)$ merely means that $p$ has to be constant on the team.
Often, due to convenience we will write $\dep(p)$ instead of $\dep(;p)$. 
Note that the hyperproperties `input determinism' now can be very easily expressed via the formula\label{pg:input-nonint}
$
\dep(i_1,\dots,i_n;o_1,\dots,o_m),
$
where $i_j$ are the (public) input variables and $o_j$ are the (public) output variables.

Problems of the form $\TSATa(\dep)$, etc.,  refer to \LTL-formulas with dependence operator $\dep$. The following proposition follows from the corresponding result for classical \LTL using downwards closure and the fact that on singleton teams dependence atoms are  trivially fulfilled.
\begin{proposition}\label{thm:TSATa-TSATs-dep}
$\TSATa(\dep)$ and $\TSATs(\dep)$ are $\PSPACE$-complete.
\end{proposition}
%\begin{proof}
%Both, asynchronous and synchronous semantics for LTL (without extensions) are downwards closed %(Lemma~\ref{lem:properties}).
%By the semantics, logics with formulas using $\dep$-atoms stay downwards closed.
%Consequently, it suffices to check satisfiability on singleton teams. On singleton teams dependence atoms are  trivially fulfilled %and consequently the claim of the theorem follows from Prop.~\ref{LTLsat}.
%\end{proof}
%
%The proof of the following theorem works also for $\TPCs(\dep)$, however, by Lemma~\ref{lem:TPCs_PSPACEhard} we already know that the problem is $\PSPACE$-hard.
%Now we consider the path checking problem for $\LTL(\dep)$.

In the following, we will show a lower bound while the matching upper bound still is open.
\begin{theorem}\label{thm:TPCa-dep}
$\TPCa(\dep)$ is $\PSPACE$-hard w.r.t.\ $\leqpm$-reductions.
\end{theorem}
\begin{proof}
%Similarly
As in the proof of Lemma~\ref{lem:TPCs_PSPACEhard}, we reduce from $\qbfval$.

Consider a given quantified Boolean formula $\exists x_1\forall x_2\cdots Q x_n\chi$, where $\chi=\bigwedge_{j=1}^m\bigvee_{k=1}^3\ell_{jk}$, $Q\in\{\exists,\forall\}$, and
 $x_1,\dots,x_n$ are exactly the free variables of $\chi$ and pairwise distinct.
We will use two traces for each variable $x_i$ (gadget for $x_i$) as shown in Figure~\ref{fig:traces-proof-tpca-dep}.
\begin{figure}
\centering
	\begin{tikzpicture}[x=.8cm]
		\node[dotwhite,label={180:$\substack{p_i\\q_i\\r_i\\s_i}$}] (a) at (0,0) {};
		
		\path[-, draw] (a) edge[loop above,->, >=stealth'] (a);
		
		\node[dotwhite, label={180:$\substack{q_i\\r_i\\\overline{p_i}}$}] (b1) at (1.5,0) {};
		\node[dot, label={0:$\substack{q_i\\s_i\\\overline{p_i}}$}] (b2) at (2.2,0) {};
		
		\path[-stealth', draw] (b1) edge [bend left] (b2);
		\path[-stealth', draw] (b2) edge [bend left] (b1);
		\node at (1,-.75) {\footnotesize everywhere $p_j,\overline{p_j}$ for $i\neq j$};
	\end{tikzpicture}
	\caption{Traces in the proof of Theorem~\ref{thm:TPCa-dep}.}\label{fig:traces-proof-tpca-dep}
\end{figure}
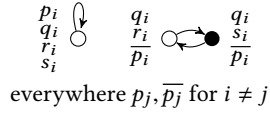

%The left trace is $(\varepsilon,\{p_i,q_i,r_i,s_i\})$ and the right is $(\varepsilon,\{q_i,r_i,\overline{p}_i\}\{q_i,s_i, \overline{p}_i\})$.
Intuitively, the proposition $p_i$ marks that the variable $x_i$ is set true while the proposition $\overline{p}_i$ marks that $x_i$ is set false, $q_i$ encodes that the gadget is used to quantify $x_i$, and $s_i,r_i$ are auxiliary propositions.
Picking the left trace corresponds to setting $x_i$ to true and picking the right trace corresponds to setting $x_i$ to false. 
In the following, we omit the $p_j$ and $\overline{p_j}$, when $j\neq i$, for readability.
Then, the team $T$ is defined as
$$
T:=\{(\varepsilon,\{p_i,q_i,r_i,s_i\}), (\varepsilon,\{q_i,r_i,\overline{p}_i\}\{q_i,s_i,\overline{p}_i\})\mid 1\leq i\leq n\}.
$$
Next, we recursively define the $\LTL(\dep)$-formula used in the reduction: $f(\chi)$ is obtained from $\chi$ by substituting every positive literal $x_i$ by $p_i$ and negated literal $\lnot x_i$ by  $\overline{p_i}$, $f(\exists x_i\psi) := \big(q_i\land\dep(p_i)\big) \lor f(\psi)\big)$, and
\begin{align*}
	%f(\chi) &:= \chi[x_1/p_1,\dots,x_n/p_n,\lnot x_1/\overline{p_1},\dots,\lnot x_n/\overline{p_n}]\\
	%f(\exists x_i\psi) &:= \big(q_i\land\dep(p_i)\big) \lor f(\psi)\\
	f(\forall x_i\psi) &:= \G\Big( \big(\dep(p_i) \land q_i \land r_i\big) \lor \big(s_i \land f(\psi)\big) \Big).
\end{align*}
In the existential quantification of $x_i$, the splitjunction requires for the $x_i$-trace-pair to put $(\varepsilon,\{p_i,q_i,r_i,s_i\})$ into the left or right subteam (of the split). 
The trace $(\varepsilon,\{q_i,r_i,\overline{p}_i\}\{q_i,s_i,\overline{p}_i\})$ has to go to the opposite subteam as $\dep(p_i)$ requires $p_i$ to be of constant value.
(Technically both of the traces could be put to the right subteam, but this logic is downwards closed and, accordingly, this allows to omit this case.)
As explained before, we existentially quantify $x_i$ by this split.
For universal quantification, the idea is a bit more involved.
%Recall the semantics of the $\G$-operator.
%\begin{align*}
%	T\amodels\G\phi &\text{ if } \forall t\in T\,\forall k_t\ge 0: \set{t[k_t,\infty) \mid t \in T}\amodels\phi.
%\end{align*}
%
To verify $T \starmodels \G \theta$, where $\G\theta = f(\forall x_i\psi)$ essentially two different teams $T'$ for which $T'\models \theta$ need to be verified. \\
%\begin{enumerate}
	%\item 
(1.) $(\varepsilon,\{p_i,q_i,r_i,s_i\}),(\varepsilon,\{q_i,r_i,\overline{p}_i\}\{q_i,s_i,\overline{p}_i\})\in T'$. In this case, $(\varepsilon,\{p_i,q_i,r_i,s_i\})$ must be put to the right subteam of the split and $(\varepsilon,\{q_i,r_i,\overline{p}_i\}\{q_i,s_i,\overline{p}_i\})$ to the left subteam, setting $x_i$ true. \\
	%\item 
(2.) $(\varepsilon,\{p_i,q_i,r_i,s_i\}),(\varepsilon,\{q_i,s_i,\overline{p}_i\}\{q_i,r_i,\overline{p}_i\})\in T'$. In this case, $(\varepsilon,\{p_i,q_i,r_i,s_i\})$ must be put to the left and $(\varepsilon,\{q_i,s_i,\overline{p}_i\}\{q_i,r_i,\overline{p}_i\})$ to the right subteam, implicitly forcing $x_i$ to be false.
%\end{enumerate}
These observations are utilised to prove that
$\langle \exists x_1\forall x_2\cdots Q x_n\chi\rangle\in\qbfval$ if and only if $\langle f(\exists x_1\forall x_2\cdots Q x_n\chi), T\rangle\in\TPCa(\dep).$ 
The reduction is polynomial time computable in the input size.
\end{proof}

The following result from Virtema talks about the validity problem of propositional team logic.
\begin{proposition}[\cite{DBLP:journals/iandc/Virtema17}]\label{prop:val-nexptime}
	Validity of propositional logic with dependence atoms is $\NEXPTIME$-complete w.r.t.\ $\leqpm$-reductions.
\end{proposition}

%In the proof of the next theorem, we use the previous result.

\begin{theorem}\label{thm:TMCa-TMCs-dep}
$\TMCa(\dep)$ and $\TMCs(\dep)$ are $\NEXPTIME$-hard w.r.t.\ $\leqpm$-reductions.
\end{theorem}
\begin{proof}
The proof of this result uses the same construction idea as in the proof of Theorem~\ref{thm:tmca-tmcs-negation}, but this time from a different problem, namely, validity of propositional logic with dependence atoms which settles the lower bound by Proposition~\ref{prop:val-nexptime}.
Due to downwards closure the validity of propositional formulas with dependence atoms boils down to model checking the maximal team in the propositional (and not in the trace) setting, which essentially is achieved by $T(\mathcal K)$, where $\mathcal K$ is the Kripke structure from the proof of Theorem~\ref{thm:tmca-tmcs-negation}.
\end{proof}

\section{LTL under Team Semantics vs.\ HyperLTL} 
 \LTL under team semantics expresses hyperproperties~\cite{DBLP:journals/jcs/ClarksonS10}, that is, sets of teams, or equivalently, sets of sets of traces. 
 Recently, \hyltl~\cite{DBLP:conf/post/ClarksonFKMRS14} was proposed to express information flow properties, which are naturally hyperproperties. For example, input determinism can be expressed as follows: every pair of traces that coincides on their input variables, also coincides on their output variables (this can be expressed in \LTL with team semantics by a dependence atom $\dep$ as sketched is Section \ref{sec:extensions}). To formalise such properties, \hyltl allows to quantify over traces. This results in a powerful formalism with vastly different properties than \LTL~\cite{DBLP:conf/stacs/Finkbeiner017}. After introducing syntax and semantics of \hyltl, we compare the expressive power of \LTL under team semantics and \hyltl.

The formulas of \hyltl are given by the grammar 
\[
\phi {} \cceq {}  \exists \pi.\phi \mid \forall \pi.\phi \mid \psi, \quad\quad \psi {}  \cceq {}  p_\pi \mid \neg \psi \mid \psi \lor \psi \mid \X \psi \mid \psi \U \psi,
\]
where $p$ ranges over atomic propositions in $\ap$ and where $\pi$ ranges over a given countable set~$\var$ of \emph{trace variables}. The other Boolean connectives and the temporal operators release~$\R$, eventually~$\F$, and always~$\G$ are derived as usual, due to closure under negation. A sentence is a closed formula, i.e., one without free trace variables.

The semantics of \hyltl is defined with respect to trace assignments that are a partial mappings~$\Pi \colon \var \rightarrow (\pow{\ap})^\omega$. The assignment with empty domain is denoted by $\Pi_\emptyset$. Given a trace assignment~$\Pi$, a trace variable~$\pi$, and a trace~$t$, denote by $\Pi[\pi \rightarrow t]$ the assignment that coincides with $\Pi$ everywhere but at $\pi$, which is mapped to $t$. Further, $\suffix{\Pi}{i}$ denotes the assignment mapping every $\pi$ in $\Pi$'s domain to $\Pi(\pi)[i,\infty) $.
For teams~$T$ and trace-assignments~$\Pi$ we define 

\begin{tabbing}
	$(T, \Pi) \hyltlmodels  \psi_1 \lor \psi_2 $ \= if \= Rechts \kill
	$(T, \Pi) \hyltlmodels p_\pi$\> if \> $p \in \Pi(\pi)(0)$,\\
	$(T, \Pi) \hyltlmodels  \neg \psi$\> if \> $(T, \Pi) \nhyltlmodels  \psi$,\\
	$(T, \Pi) \hyltlmodels  \psi_1 \lor \psi_2 $\> if \> $(T, \Pi) \hyltlmodels  \psi_1$ or $(T, \Pi) \hyltlmodels  \psi_2$,\\
	$(T, \Pi) \hyltlmodels  \X \psi$\> if \> $(T,\suffix{\Pi}{1}) \hyltlmodels  \psi$,\\
	$(T, \Pi) \hyltlmodels  \psi_1 \U \psi_2$\> if \> $\exists k \ge 0: (T,\suffix{\Pi}{k}) \hyltlmodels  \psi_2$ and
	$\forall 0 \le k' < k: (T,\suffix{\Pi}{k'}) \hyltlmodels  \psi_1$, \\
	$(T, \Pi) \hyltlmodels  \exists \pi.\psi$\> if \> $\exists t \in T: (T,\Pi[\pi \rightarrow t]) \hyltlmodels  \psi$, and \\
	$(T, \Pi) \hyltlmodels  \forall \pi.\psi$\> if \> $\forall t \in T: (T,\Pi[\pi \rightarrow t]) \hyltlmodels  \psi$. 
\end{tabbing}

We say that $T$ satisfies a sentence~$\phi$, if $(T, \Pi_\emptyset) \hyltlmodels  \phi$, and write $T \hyltlmodels  \phi$.
The semantics of \hyltl are synchronous, i.e., the semantics of the until refers to a single $k$. Accordingly, one could expect that \hyltl is closer related to \LTL under synchronous team semantics than to \LTL under asynchronous team semantics. In the following, we refute this intuition. 

Formally, a \hyltl sentence~$\phi$ and an \LTL formula~$\phi'$ under synchronous (asynchronous) team semantics are equivalent, if for all teams~$T$: $T \hyltlmodels \varphi$ if and only if $T \smodels \phi'$ ($T \amodels \phi'$).
In the following, let $\forall$-\hyltl denote that set of \hyltl sentences of the form~$\forall \pi.\, \psi$ with quantifier-free~$\psi$, i.e., sentences with a single universal quantifier.

  \begin{theorem}
  \label{theorem_hyltlvsteam}

  	\begin{enumerate}
  		\item\label{theorem_hyltlvsteam_teamweakerthanshyltl} No \LTL-formula under synchronous or asynchronous team semantics is equivalent to $\exists\pi.p_\pi$.
  			%$\exists\pi.\, p_\pi$ has no equivalent \LTL formula under neither team semantics.
  		\item\label{theorem_hyltlvsteam_hyltlweakerthansynchrteam} No \hyltl sentence is equivalent to $\F p$ under synchronous team semantics. 
  		\item\label{theorem_hyltlvsteam_hyltlequalsasynchrteam} \LTL under asynchronous team semantics is as expressive as $\forall$-\hyltl.
  	\end{enumerate}
  \end{theorem}

  \begin{proof}
	\ref{theorem_hyltlvsteam_teamweakerthanshyltl}. Consider $T = \set{\emptyset^\omega, \set{p} \emptyset^\omega}$. We have $T \hyltlmodels  \exists\pi.p_\pi$. Assume there is an equivalent \LTL formula under team semantics, call it $\phi$. Then, $T \starmodels \phi$ and thus $\set{\emptyset^\omega} \starmodels \phi$ by downwards closure. Hence, by equivalence, $\set{\emptyset^\omega} \hyltlmodels  \exists\pi.p_\pi$, yielding a contradiction.

  \ref{theorem_hyltlvsteam_hyltlweakerthansynchrteam}. Bozzelli et al.\ proved that the property encoded by $\F p$ under synchronous team semantics cannot be expressed in \hyltl~\cite{BozzelliMP15}.

	\ref{theorem_hyltlvsteam_hyltlequalsasynchrteam}. Let $\varphi$ be an \LTL-formula and define $\varphi_h \dfn \forall \pi.\phi'$, where $\varphi'$ is obtained from $\varphi$ by replacing each atomic proposition~$p$ by $p_\pi$. Then, due to singleton equivalence, $T \amodels \varphi$ if and only if $T \hyltlmodels \varphi_h$.
	For the other implication, let $\varphi = \forall \pi.\psi$ be a \hyltl sentence with quantifier-free $\psi$ and let $\psi'$ be obtained from $\psi$ by replacing each atomic proposition~$p_\pi$ by $p$. Then, again due to the singleton equivalence, we have $T \hyltlmodels \varphi$ if and only if $T \amodels \psi'$.
    \end{proof}
  
Note that these separations are obtained by very simple formulas, and are valid for $\LTL(\dep)$ formulas, too. 
In particular, the \hyltl formulas are all negation-free.
   
  \begin{corollary}
  \hyltl and \LTL under synchronous team semantics are of incomparable expressiveness and
  \hyltl is strictly more expressive than \LTL under asynchronous team semantics.
  \end{corollary}

\section{Conclusion}
We introduced synchronous and asynchronous team semantics for linear temporal logic $\LTL$, studied complexity and expressive power of related logics, and compared them to $\hyltl$.
We concluded that \LTL under team semantics is a valuable logic which allows to express relevant hyperproperties and complements the expressiveness of \hyltl while allowing for computationally simpler decision problems.
We conclude with some directions of future work and open problems.
\begin{enumerate}
\item We showed that some important properties that cannot be expressed in $\hyltl$ (such as uniform termination) can be expressed by $\LTL$-formulas in synchronous team semantics. Moreover input determinism can be expressed in $\LTL(\dep)$. What other important and practical hyperproperties can be expressed in $\LTL$ under team semantics? What about in its extensions with dependence, inclusion, and independence atoms, or the contradictory negation.
\item We showed that with respect to expressive power $\hyltl$ and $\LTL$ under synchronous team semantics are incomparable. What about the extensions of $\LTL$ under team semantics? For example the $\hyltl$ formula $\exists \pi.p_\pi$ is expressible in $\LTL(\sim)$. Can we characterise the expressive power of relevant extensions of team $\LTL$ as has been done in first-order and modal contexts?
\item We studied the complexity of path-checking, model checking, and satisfiability problems of team $\LTL$ and its extensions with dependence atoms and the contradictory negation. Many problems are still open: Can we show matching upper bounds for the hardness results of Section \ref{sec:extensions}? What is the complexity of $\TMCs$ when splitjunctions are allowed?
What happens when $\LTL$ is extended with inclusion or independence atoms?
\item Can we give a natural team semantics to $\CTL^*$ and compare it to $\mathrm{HyperCTL}^*$ \cite{DBLP:conf/post/ClarksonFKMRS14}?
\end{enumerate}
%
%%%%%%%%%%%%
\begin{comment}
\section{Old Conclusion}
For all three considered decision problems the complexity of the asynchronous semantics coincides with the one of \LTL formulas under classical semantics.
However, the study of the complexities for the synchronous variant of the semantics gave compelling new insights into the interplay of formal language theoretic formalisms and team semantics.
The path checking problem for synchronous semantics is $\PSPACE$-complete.
We proved a $\PSPACE$ upper bound for the splitjunction-free fragment of the model checking problem.
We observed that adding dependence atoms yields increase in terms of complexity presumably only for the model-checking problem.
Extending the language by the contradictory negation yields a more dramatic effect: often we reach alternating exponential time with polynomially many alternations.
Allowing FO-definable generalised atoms does not increase the complexity of the path checking problem in the synchronous case.

Finally, for \LTL under team semantics as well as for \hyltl there exist formulas which cannot be expressed in the other logic.
Concluding, \LTL under team semantics is a valuable logic which allows to express relevant hyperproperties and thereby complements the expressiveness of \hyltl while allowing for computationally simpler decision problems.
\end{comment}
%%%%%%%%